\documentclass[12pt]{article}%
\usepackage{amsfonts}
\usepackage{amssymb}
\usepackage{amsmath}
\usepackage{float}
\usepackage{color}
\usepackage{graphicx}
\usepackage[round,authoryear,comma]{natbib}
\usepackage{hyperref}%
\usepackage{soul} 
\usepackage{multirow} 
\usepackage{subcaption}
\usepackage{empheq} 
\usepackage{enumitem} 
\usepackage[normalem]{ulem} 
\usepackage{cancel} 
\usepackage{tikz}
\usetikzlibrary{arrows.meta}

\providecommand{\U}[1]{\protect\rule{.1in}{.1in}}
\restylefloat{figure}
\newtheorem{theorem}{Theorem}

\newtheorem{corollary}{Corollary}
\newtheorem{definition}{Definition}

\newtheorem{lemma}{Lemma}

\newtheorem{proposition}{Proposition}
\newtheorem{remark}{Remark}

\newtheorem{example}{Example}

\let\oldexample\example
\renewcommand{\example}{\oldexample\normalfont}
\newenvironment{proof}[1][Proof]{\noindent \textbf{#1:} }{\hfill \rule{0.5em}{0.5em}}

\parskip=\medskipamount
\begin{document}

\title{Compensation-based risk-sharing}
\author{Jan Dhaene\thanks{jan.dhaene@kuleuven.be, KU Leuven, Belgium} \and Atibhav Chaudhry\thanks{atibhav.chaudhry@kuleuven.be, University of Melbourne, Australia and KU Leuven, Belgium} \and Ka Chun Cheung\thanks{kccg@hku.hk, University of Hong Kong, Hong Kong} \and Austin Riis-Due\thanks{austincarter.riis-due@uwaterloo.ca, University of Waterloo, Canada}}

\maketitle
\begin{abstract}
This paper studies the mathematical problem of allocating payouts (compensations) in an endowment contingency fund using a risk-sharing rule that satisfies full allocation. Besides the participants, an administrator manages the fund by collecting ex-ante contributions to establish the fund and distributing ex-post payouts to members. Two types of administrators are considered. An `active’ administrator both invests in the fund and receives the payout of the fund when no participant receives a payout. A `passive’ administrator performs only administrative tasks and neither invests in nor receives a payout from the fund. We analyze the actuarial fairness of both compensation-based risk-sharing schemes and provide general conditions under which fairness is achieved. The results extend earlier work by Denuit and Robert (2025) and Dhaene and Milevsky (2024), who focused on payouts based on Bernoulli distributions, by allowing for general non-negative loss distributions. 
\end{abstract}


\section{Centralized vs. decentralized risk-pooling}

Consider a group of $n$ individuals observed at time $0$, each of them exposed
to a non-negative random loss at time $1$.\ This loss can be related to a
well-defined peril (e.g., hospitalization-related or critical illness-related
expenses), or it can be expressed as a deterministic claim payment contingent
on the occurrence of a well-defined event (e.g., a deterministic payment
contingent on death or survival).

In the \textit{classical insurance} approach, individuals involve an
insurer and each of them buys an insurance contract at time $0$, which
entitles them to their respective observed losses at time $1$. Every
policyholder is compensated ex-post (at time $1$) for his experienced loss.
In return for this coverage, the insurer charges each policyholder an ex-ante insurance premium (at time $0$). In such an insurance coverage, the aggregate risk
(randomness) of the insurance portfolio is taken over by the insurer. This
risk transfer is possible if the pool consists of a sufficient number of
mutually independent and homogeneous risks, with premiums being calculated in
a conservative way.\ In addition, the insurer sets up solvency capital for the
case that the collected premiums turn out to be insufficient to cover the
guaranteed claims.\ Premiums and solvency capital are chosen so that the
probability of the event that the sum of all accumulated premiums and solvency
capital exceeds the aggregate claims of the insurance portfolio is
sufficiently large (e.g., 99.5\%). Classical insurance as described above is a
form of \textit{centralized risk transfer}, meaning that it is a risk-transfer
mechanism in which individual losses faced by policyholders are transferred
to a central insurer who guarantees that losses will be paid. In exchange, policyholders pay a risk premium and compensate the insurer’s shareholders by providing a return on the solvency capital they set aside to maintain the insurer’s guarantee.

Instead of a transfer of the aggregate risk of the pool to a central insurer
(guarantor), individuals can opt for a so-called \textit{decentralized
risk-sharing} approach, where individuals do not transfer the aggregate losses
to a guarantor but keep the aggregate risk within the pool, without
generating or creating any insolvency risk. Examples of such approaches can be found in Abdikerimova and Feng (2022) and the references
therein. One way to achieve this
goal is that the premiums at time $0$ of the classical centralized approach are
replaced by time $1$ contributions: Each participant in the risk-sharing pool
is fully compensated  for his loss at time $1$, but in return he pays a
contribution to the pool at time $1$.\ These contributions are time $1$ measurable random
variables, chosen at time $0$.\ The risk-sharing scheme is set up
such that the sum of all contributions paid at time $1$ by the participants is
exactly equal to the sum of all losses covered by the pool.\ This constraint
is called the full allocation condition.\ In other words, participants
contribute at time $1$, by sharing total losses once they have
been observed. Following Dhaene and Milevsky (2024), we call such a
decentralized approach \textit{contribution-based risk-sharing}.\ 

A simple example of such a risk-sharing approach is the uniform risk-sharing
rule, where each participant's contribution is set equal to the observed
aggregate losses of the pool, divided by the number of participants.\ Other
examples of contribution-based risk-sharing are the conditional-mean risk
sharing scheme, introduced in the actuarial literature in Denuit and Dhaene
(2012), and the quantile risk-sharing scheme introduced in Denuit, Dhaene and
Robert (2022).\ The properties of contribution-based
risk-sharing rules have been investigated in detail in Denuit, Dhaene and
Robert (2022) and Denuit, Dhaene, Ghossoub and Robert (2025), among others. An
axiomatic characterization of the conditional mean risk-sharing rule is given
in Jiao, Kou, Liu and Wang (2023), while an axiomatic characterization of the
quantile risk-sharing rule is considered in Dhaene, Cheung, Robert and Denuit
(2025). Axiomatic characterizations of some simple risk-sharing rules,
including the uniform rule, are presented in Dhaene, Kazzi and Valdez (2025).\ 

Decentralized risk-sharing can also be constructed in another way.\ Indeed,
suppose again that all participants are exposed to a random non-negative loss at time $1$.\ Each of them invests an amount at time $0$ to set up a so-called `endowment
contingency fund'. At time $1$, the total fund value
is shared among all participants. The relative part that each participant will
receive is a time $1$ observable random variable,\ determined at time $0$ as a
well-defined function of the claims and eventually also of other
information that will be observable at time $1$.\ Solvency is guaranteed by the full allocation condition which states
that the sum of all payments to the participants at time $1$ is equal to the fund
value at time $1$.\ The aim of an `endowment contingency fund' is to provide
participants with a cheaper and effective protection, compared to commercial
insurance.\ The term `endowment' indicates that it is an investment portfolio
with initial capital deriving from cash inflows, whereas the term
`contingency' means that the payments out of this fund are contingent on the
realization of certain random events.\ For each individual participant, the
coverage ratio (i.e., compensation over claim) depends on the initial
investments and the observed losses of all participants. A participant may
not receive his observed loss, i.e., he might not be fully compensated for his
loss, or even receive more than his loss.\ In this setting, the time $1$
payments can be considered as a kind of compensation for the occurred
losses.\ Therefore, Dhaene and Milevsky (2024) call the time $1$ payments from
the fund to the participants the compensations, and they baptize this
decentralized risk-sharing approach \textit{compensation-based risk-sharing}.

As a simple example of the compensation-based approach, consider the
risk-sharing scheme where at time $1$, each
participant receives a compensation which is proportional to his observed
loss.\ The proportionality factor, which is observable at time $1$, is assumed to be equal for all participants and follows from the full allocation condition.\ 

Special cases of compensation-based risk-sharing schemes have been
investigated in several actuarial papers, including Denuit and Robert (2025), Dhaene and Milevsky (2024) and Bernard, Feliciangeli and Vanduffel (2025), where only two-point distributed losses (indicator random variables) have
been considered.\ In the current paper, we generalize the approach set up in these papers to include general losses and further explore properties of general compensation-based risk-sharing schemes. The only paper of which we are aware that follows our general approach is Denuit and Robert (2026). Our generalization broadens the applicability of such schemes to a wider range of insurance types with random claim severity, for example, homeowners and automobile insurance. By presenting some simple examples with closed-form solutions, we aim to make the underlying principles more transparent.

We analyze the roles of an active and a passive administrator in such schemes, thereby unifying within a single setting the active-administrator setting of Dhaene and Milevsky (2024) and the passive-administrator setting of Denuit and Robert (2025). The active administrator is a simple capacity provider to the pool. Capital markets often exhibit a desire to participate in insurance risk as it is uncorrelated with other financial assets (Canabarro, Finkemeier, Anderson, Bendimerad, 2000). For funds with few participants, the active administrator plays a critical role for actuarial fairness and preventing the fund from becoming non-viable. In practice, any investor or (re-)insurer may participate as an active administrator. On the other hand, the passive administrator performs only administrative tasks and neither invests in nor receives a payout from the fund.

We present a simple utility-based comparison of compensation-based risk-sharing, classical insurance and self-insurance. For a more comprehensive analysis using utility theory, we refer the interested reader to Bernard, Feliciangeli and Vanduffel (2025), as well as to the ongoing work of Chen, Cheung, Dhaene, and Yam (2026), who develop a welfare-optimal framework for compensation-based risk-sharing in tontine funds with heterogeneous investors.

All random variables considered in this paper are defined on the probability space $\left(  \Omega,\mathcal{F}%
,\mathbb{P}\right)  $. The set of all non-negative random variables on $\left(
\Omega,\mathcal{F},\mathbb{P}\right)  $ is denoted by $L_{+}^{0}$. Throughout
this paper, the term `positive' is used for `strictly larger than zero'. The
set of non-negative real numbers is denoted by $%
\mathbb{R}
^{+}$.

\section{Compensation-based risk-sharing with an active administrator}

In this paper, we introduce and investigate a general type of fully-funded
risk-sharing (further abbreviated as RS) mechanisms, where $n$ participants
decide to mutually invest and set up a so-called `contingency endowment
fund'.\ At the beginning of the investment period, each participant $i$ makes
an initial (non-negative) investment $\pi_{i}$ in the fund. Throughout this paper, we always implicitly assume that at least one of these investments is positive, which means that
for at least one participant $i$, $i=1,2,\ldots,n$, one has that $\pi_{i}>0$.
The beginning of the investment period is denoted by time $0$, and also
referred to as `now'.\ The end of the investment period is denoted by time $1$, also referred to as `the end of the year'; however, in general, time-$1$ may indicate the endpoint of a time interval substantially longer than one year. In practice, these investments will be supplemented with fees to cover expenses, but we assume that the investments $\pi_{i}$ are net investments, after fees for expenses have been paid.\ In this case, we do not further have to take these fees into account in our analysis.\ 

Our objective is to investigate fair methods for the participants to divide
the total initial investment among themselves at time $1$. To be more precise, at time $0$,
the time $1$ observable non-negative random variables $W_{i}$ are chosen, where $W_{i}$
stands for the part of the total fund value available at time $1$ that will be
attributed to the participant $i$ at this time. Following Dhaene and Milevsky
(2024), we call these amounts to be paid at time $1$ the `compensations' to the
participants. In general, there is the possibility that the time $1$
realizations of all compensations $W_{i}$, $i=1,2,\ldots,n$, are equal to $0$.
At time $0$, we have to clearly specify what happens to the fund's proceeds in
this case.\ 

Apart from the $n$ participants, there is another agent, denoted by $n+1$ and
called the administrator. The role of the administrator is to collect the amounts $\pi_{i}$\ at
time $0$, to invest them, and to distribute the compensations $W_{i}$ to the
$n$ participants at time $1$. Seen from time $0$, the compensations to the $n$ participants are random variables defined and agreed upon by the participants at that time. As mentioned above, it might happen that all
compensations to the participants are equal to $0$.\ We assume that in case
this situation occurs, the `active' administrator receives the full proceeds of the
fund.\ On the other hand, in case the time $1$ realization of the compensation
of at least one participant is positive, the administrator does not receive
anything. The administrator's compensation at time $1$ is denoted by
$W_{n+1}$.\ We assume that the administrator also contributes an initial
(non-negative) investment $\pi_{n+1}$ to the fund, in return for receiving the
total proceeds of the fund when $W_{i}=0,$ $i=1,2,\ldots,n$.

When referring to the `$n+1$ agents' of a compensation-based RS scheme, we
mean the $n$ participants and the administrator.\ The sum of the initial
investments of all agents, that is, $\sum_{j=1}^{n+1}\pi_{j}$, is equal to the total
value of the fund at the beginning of the investment period. For simplicity,
we assume a zero interest rate and no expenses, but our results can easily be generalized to
the case of a (deterministic) investment return and expenses by assuming $\sum_{j=1}^{n+1} \pi_j$ is the accumulated value of the investments net of expenses.\ Under this assumption, the
time $1$ value of the investment fund is equal to $\sum_{j=1}^{n+1}\pi_{j}$.
Notice that our implicit assumption that at least one of the participants'
investments is positive implies that the time $1$ value of the fund is always positive.

The active administrator's time $1$ claim $W_{n+1}$ on the fund can be
expressed as follows:
\begin{equation}
W_{n+1}=\left\{
\begin{array}
[c]{ll}%
\sum_{j=1}^{n+1}\pi_{j} & :W_{1}=W_{2}=\cdots=W_{n}=0\\
0 & :\text{ otherwise \ \ \ \ \ \ \ \ \ \ \ \ \ \ \ \ \ \ \ \ }%
\end{array}
\right.  \label{C3}%
\end{equation}
This means that the administrator receives the full proceeds of the fund if
and only if all other participants receive zero compensation.\ Taking into
account that all compensations are non-negative, the event
$^{\prime}W_{1}=W_{2}=\ldots=W_{n}=0^{\prime}$ is equivalent to the event
$^{\prime}\sum_{j=1}^{n}W_{j}=0^{\prime}$. The random variables $\sum_{j=1}^{n}W_{j}$
and $W_{n+1}$ exhibit a dependency structure which is a special case of
`countermonotonicity', called `mutual exclusivity', meaning that both random variables
are non-negative, with one of them being positive implying that the other one
is equal to $0$.\ The concept of `mutual exclusivity' is considered in several
papers in the actuarial literature, see for example, Dhaene and Denuit (1999), Cheung
and Lo (2014) and Lauzier, Lin and Wang (2024).\ 

Introducing the random variable $P_{n+1}$, with $0\leq P_{n+1}\leq1$, for the proportion
of the aggregate investment that will be attributed to the administrator,
the compensation $W_{n+1}$ can be expressed as follows:
\begin{equation}
W_{n+1}=\left(  \sum_{j=1}^{n+1}\pi_{j}\right)  \times P_{n+1}, \label{C4}%
\end{equation}
with
\begin{equation}
P_{n+1}=\left\{
\begin{array}
[c]{ll}%
1 & :W_{1}=W_{2}=\cdots=W_{n}=0\\
0 & :\text{otherwise \ \ \ \ \ \ \ \ \ \ \ \ \ \ \ \ \ \ \ }%
\end{array}
\right.
\end{equation}

Introducing the random variable $P_{i}$, with $0\leq$ $P_{i}\leq1$, for the random
proportion of the aggregate investment attributed to participant $i$, we have
that the compensations of the participants can be expressed as%
\begin{equation}
W_{i}=\left(  \sum_{j=1}^{n+1}\pi_{j}\right)  \times P_{i}\qquad
i=1,2,\ldots,n. \label{C6a}%
\end{equation}
Figure 1 illustrates the cash flows for each agent at times 0 and 1.

\begin{figure}[ht]
  \centering
  \begin{tikzpicture}[
      font=\small,
      every node/.style={text=black},
      every path/.style={draw=black}
    ]
    \draw[thick] (0,0) -- (11,0);
    \draw[thick] (0,-0.12)  -- (0,0.12);
    \draw[thick] (11,-0.12) -- (11,0.12);
    \node[above=4pt] at (0,0.12)  {Time $0$};
    \node[above=4pt] at (11,0.12) {Time $1$};
    \node[anchor=north, align=center, inner ysep=8pt] at (0,-0.12) {%
      Invest $\pi_i$ into the fund%
    };
    \node[anchor=north, align=center, inner ysep=8pt] at (9.5,-0.12) {%
      Receive $W_i = \bigl(\sum_{j=1}^{n+1} \pi_j\bigr) \times P_i$ from the fund%
    };
  \end{tikzpicture}
  \caption{Cash flows for agent $i$, where $i = 1, 2, \ldots, n+1$.}
  \label{fig:timeline}
\end{figure}

Taking into account (\ref{C6a}), the random relative compensation $P_{n+1}$ attributed to
the administrator can be expressed as follows:
\begin{equation}
P_{n+1}=\left\{
\begin{array}
[c]{ll}%
1 & :P_{1}=P_{2}=\cdots=P_{n}=0\\
0 & :\text{otherwise, \ \ \ \ \ \ \ \ \ \ \ \ \ \ \ \ \ \ \ \ }%
\end{array}
\right.  \label{C6b}%
\end{equation}
or equivalently,
\begin{equation}
P_{n+1}={ 1}\left(  \sum_{j=1}^{n}P_{j}=0\right)  , \label{C6c}%
\end{equation}
where ${ 1}\left(  A\right)  $ stands for the indicator function, which
is $1$ if the event $A$ occurs and $0$ otherwise.\ This expression for $P_{n+1}$ also shows that $\sum_{j=1}^{n}P_{j}$ and $P_{n+1}$ are mutually exclusive.

Hereafter, we always assume that at time $1$, the total amount of the
available funds is fully distributed to the participants and the
administrator, that is,
\begin{equation}
\sum_{j=1}^{n+1}P_{j}=1. \label{C11a}%
\end{equation}
This condition is called the full allocation condition for the
compensation-based RS scheme.\ We also always implicitly assume that
\begin{equation}
0<\Pr\left[  P_{n+1}=0\right]  <1, \label{C30}%
\end{equation}
or equivalently,
\begin{equation}
0<\Pr\left[  P_{n+1}=1\right]  <1. \label{C30a}%
\end{equation}
The assumption (\ref{C30}) means that the events `at least one participant
receives a non-zero compensation' and `all participants receive a
zero-compensation' have a positive probability. Taking into account that
$P_{n+1}$ is Bernoulli distributed, from (\ref{C30a}) we find that
\begin{equation}
0<E\left[  P_{n+1}\right]  =\Pr\left[  P_{n+1}=1\right]  <1. \label{C31}%
\end{equation}
From (\ref{C4}) it follows then that
\begin{equation}
0<E\left[  W_{n+1}\right]  <\sum_{j=1}^{n+1}\pi_{j}. \label{C31a}%
\end{equation}
\bigskip Finally, from (\ref{C11a}) and (\ref{C30}) we have that
\begin{equation}
0<\sum_{k=1}^{n}E\left[  P_{k}\right]  =\Pr\left[  P_{n+1}=0\right]  <1.
\label{C34a}%
\end{equation}

We will call the administrator as described above an `active administrator',
where `active' means that he is the `owner' of the compensation $W_{n+1}%
$. The role of the active administrator was introduced in this setting by Dhaene and Milevsky (2024). Further in the paper, we will also introduce the `passive administrator'
whose compensation $W_{n+1}$ will be redistributed to the agents.

Let us now introduce the (deterministic) investment vector $\boldsymbol{\pi}$,
which is defined by
\begin{equation}
\boldsymbol{\pi}=\left(  \pi_{1},\pi_{2},\ldots,\pi_{n+1}\right)  ,
\label{C7a}%
\end{equation}
as well as the time $1$ measurable compensation vector $\mathbf{W}$, defined
by
\begin{equation}
\mathbf{W}=\left(  W_{1},W_{2},\ldots,W_{n+1}\right)  , \label{C7b}%
\end{equation}
and the time $1$ measurable relative compensation vector $\mathbf{P}$,
defined by%
\begin{equation}
\boldsymbol{P}=\big(P_{1},P_{2},\ldots,P_{n+1}\big). \label{C7}%
\end{equation}
Taking into account the introduced vector notations, we can rewrite (\ref{C4})
and (\ref{C6a}) in the following way:%
\begin{equation}
\boldsymbol{W}=\left(  \sum_{j=1}^{n+1}\pi_{j}\right)  \times\boldsymbol{P}.
\label{C6}%
\end{equation}

In the following definition, we introduce the set of all relative compensation
vectors for the $n+1$ agents.\ \ 

\begin{definition}
The set $\mathcal{R}_{n+1}$\ is defined by
\begin{equation}
\mathcal{R}_{n+1}=\left\{  \left(  P_{1},P_{2},\ldots,P_{n+1}\right)  \in
(L_{+}^{0})^n\mid\sum_{j=1}^{n+1}P_{j}=1\text{ and }P_{n+1}={1}\left(
\sum_{j=1}^{n}P_{j}=0\right)  \right\}  . \label{C6d}%
\end{equation}

\end{definition}

The set $\mathcal{R}_{n+1}$ consists of all $\left(  n+1\right)
$-dimensional relative compensation vectors. Any $\boldsymbol{W}=\left(  \sum
_{j=1}^{n+1}\pi_{j}\right)  \times\boldsymbol{P}$, where $\boldsymbol{P}%
\in\mathcal{R}_{n+1}$ is a reallocation (or redistribution) of the time $1$
value of the fund between the $n$ participants and the administrator, such
that $\sum_{j=1}^{n}W_{j}$ and $W_{n+1}$ are mutually exclusive.\ The latter
condition means that the events `at least one $W_{i}$ is positive' and
`$W_{n+1}$ is positive' are mutually exclusive.\ 

Compensation-based RS is a two-stage process.\ At time $0$, any agent makes an initial investment $\pi_{i}$\ and their
aggregate investment $\left(  \sum_{j=1}^{n+1}\pi_{j}\right)  $ is reallocated
by transforming $\boldsymbol{\pi}$ into random vector $\boldsymbol{W}$ defined
in (\ref{C6}), with $\boldsymbol{P}\in\mathcal{R}_{n+1}$. At time $1$, the participants and the administrator receive
the respective realizations of the compensations that were attributed to
them.$\ $
Hereafter, we denote the compensation $W_{i}$ by $W_{i}\left[
\boldsymbol{\pi}\right]  $ and the compensation vector $\boldsymbol{W}$ by
$\boldsymbol{W}\left[  \boldsymbol{\pi}\right]  $ in case we want to emphasize
that the underlying investment vector is $\boldsymbol{\pi}$. Similarly, we
write the relative compensation $P_{i}$ by $P_{i}\left[  \boldsymbol{\pi}\right]  $
and the relative compensation vector $\boldsymbol{P}$ by $\boldsymbol{P}\left[
\boldsymbol{\pi}\right]  $ to denote their dependence on $\boldsymbol{\pi}$.
When the meaning is clear, we will omit the $\left[  \boldsymbol{\pi}\right]
$ in the notation.\ \ 

We are now ready to define a compensation-based RS scheme with an active administrator.\

\begin{definition}
A compensation-based risk-sharing scheme for a given group of $n$ participants
and an active administrator is a pair $\left(  \boldsymbol{\pi},\boldsymbol{P}%
\right)  $, where $\boldsymbol{\pi}=\left(  \pi_{1},\pi_{2},\ldots,\pi
_{n+1}\right)  $ is the initial investment vector, while $\boldsymbol{P}%
=\left(  P_{1},P_{2},\ldots,P_{n+1}\right)  $ is a relative compensation
vector, that is $\boldsymbol{P}\in\mathcal{R}_{n+1}$.\ Moreover, the
compensations attributed to the $n+1$ agents are expressed by the compensation
vector $\boldsymbol{W}$, which is defined by
\begin{equation}
W_{i}=\left(  \sum_{j=1}^{n+1}\pi_{j}\right)  \times P_{i}\qquad
i=1,2,\ldots,n+1. \label{C1aa}%
\end{equation}

\end{definition}

A compensation-based RS \textit{scheme} with an active administrator is
specified by its investment vector $\boldsymbol{\pi}$ and its relative
compensation vector $\boldsymbol{P}$ (or its compensation vector
$\boldsymbol{W}$). The administrator is called `active' in the sense that he
is entitled to a random compensation at time $1$.\ Further in this paper, we
will also consider a `passive' administrator, who in not entitled to any
random compensation at time $1$.\ We remark that investments are known at time $0$, whereas the random compensations remain
unknown until they become observable at time $1$.\ As $\boldsymbol{P}\in\mathcal{R}_{n+1}$, a compensation-based RS scheme with an
active administrator satisfies the following \textbf{full allocation
condition}:
\begin{equation}
\sum_{j=1}^{n+1}W_{j}=\sum_{j=1}^{n+1}\pi_{j}, \label{C1}%
\end{equation}
which means that such an RS scheme is a fully funded system and has no
insolvency issues: the total amount that will be distributed at time $1$ will
also be available at that time.

After having introduced compensation-based RS \textit{schemes}, we can now
define a compensation-based RS \textit{rule} as an appropriate set of compensation-based RS schemes.

\begin{definition}
A compensation-based risk-sharing rule for a given group of $n$ participants
and an active administrator is a mapping $\boldsymbol{P}:\left(
\mathbb{R}
^{+}\right)  ^{n+1}\rightarrow\mathcal{R}_{n+1}$ which transforms any
investment vector $\boldsymbol{\pi}$ in $\left(
\mathbb{R}
^{+}\right)  ^{n+1}$ into a relative compensation vector $\boldsymbol{P}%
\left[  \boldsymbol{\pi}\right]  $:%
\begin{equation}
\boldsymbol{\pi}\in\left(
\mathbb{R}
^{+}\right)  ^{n+1}\rightarrow\boldsymbol{P}\left[  \boldsymbol{\pi}\right]
\in\mathcal{R}_{n+1}. \label{C1c}%
\end{equation}

\end{definition}

In this paper, when we consider a `RS scheme' or a `RS rule', we always mean a
`compensation-based RS scheme' or a `compensation-based RS rule'.\ 

\section{Some examples}

In this section, we illustrate the concept of compensation-based RS with an
active administrator with some examples. 

\begin{example}
Suppose that each participant $i$ of a group of $n$ individuals is exposed to
a (random) non-negative loss $X_{i}$ at time $1$.\ The $n$ participants decide
to share the risk related to these losses amongst themselves.\ Therefore, they
appoint an administrator, called agent $n+1$, and set up an RS scheme $\left(
\boldsymbol{\pi},\boldsymbol{P}\right)  $, with%
\begin{equation}
P_{i}=\frac{X_{i}}{\sum_{j=1}^{n+1}X_{j}},\qquad i=1,2,\ldots,n+1, \label{C35}%
\end{equation}
with $X_{n+1}$ defined by
\begin{equation}
X_{n+1}=1\left(  \sum_{j=1}^{n}X_{j}=0\right)  . \label{C35b}%
\end{equation}
It is easy to see that $X_{n+1}$ and $\sum_{j=1}^{n}X_{j}$ are mutually
exclusive.\ This mutual exclusivity property implies that the denominator in
(\ref{C35}) is never equal to $0$, so that the $P_i$'s are always well-defined.\ Obviously, the mutual exclusivity of $\left(  X_{n+1},\sum_{j=1}%
^{n}X_{j}\right)  $ is equivalent to the mutual exclusivity of $\left(
P_{n+1},\sum_{j=1}^{n}P_{j}\right)  $, which is equivalent to the mutual
exclusivity of $\left(  W_{n+1},\sum_{j=1}^{n}W_{j}\right)  $.\newline The
participants in the RS scheme with relative compensation vector defined by (\ref{C35})
share the proceeds of the fund proportionally, where each participant's
proportion is equal to the proportion that he contributes to aggregate claims
$\sum_{j=1}^{n+1}X_{j}$.\ On the other hand, the administrator receives the
full proceeds of the fund in case all participants experience a zero loss,
whereas he receives nothing in the other case. Typically, insurance losses
have a strictly positive probability mass at zero, implying that our
assumption that $\Pr\left[  P_{n+1}=1\right]  >0$ is reasonable.\newline
Remark that the compensations $W_{i}$ defined via the relative compensations
(\ref{C35}) can be expressed as follows:
\begin{equation}
W_{i}=\frac{\sum_{j=1}^{n+1}\pi_{j}}{\sum_{j=1}^{n+1}X_{j}}\times X_{i},\qquad
i=1,2,\ldots,n+1. \label{C35a}%
\end{equation}
This means that the RS scheme $\left(  \boldsymbol{\pi},\boldsymbol{P}\right)
$ is such that each participant $i$ is compensated the same time $1$
observable proportion of his loss $X_{i}$, while the administrator receives the full proceeds of the fund in case each participant has a zero claim. This seems to be a reasonable way to distribute the total fund over the $n+1$ agents, provided each
agent's initial investment $\pi_{i}$ is `reasonable' or `fair'.\ In Section 4, we
introduce and investigate `actuarially fair' initial investments.\\

\noindent As with classical insurance, a compensation-based risk-sharing scheme is susceptible to moral hazard: participants have an incentive to exaggerate or fabricate losses in order to obtain a larger share of the fund's proceeds. To mitigate this, the $X_i$ can be interpreted as adapted claim sizes, for example, proportionally reduced by a factor $\alpha \in (0,1)$ for each participant $i$. In this way, standard techniques from classical insurance, such as proportional coverage, deductibles, and no-claims discounts, can be adapted to the compensation-based risk-sharing setting in order to reduce moral hazard.$\vartriangleleft$
\end{example}

\begin{example}
Consider a group of $n$ participants.\ Each of them may experience a
particular event of a given type in the observation period $\left[
0,1\right]  $.\ Possible events include the participant's death, survival,
being hospitalized, being diagnosed with a critical illness, etc.\ For
simplicity, hereafter we will assume that all predefined events are `survival
to time $1$', but any other choice for the predefined events is possible.\ For
each participant $i$, we introduce the indicator variable $I_{i}$, which is
defined by
\begin{equation}
I_{i}=\left\{
\begin{array}
[c]{ll}%
1 & :i\text{ survives until time }1\\
0 & :i\text{ dies before time }1
\end{array}
\right.  \label{C5a}%
\end{equation}

The participants appoint an active administrator and attach the following
indicator variable $I_{n+1}$ to him:
\begin{equation}
I_{n+1}=%
{\displaystyle\prod\limits_{j=1}^{n}}
\left(  1-I_{j}\right)  .\ \label{C5b}%
\end{equation}
We further introduce the notation $p_i = \Pr \left[ I_i=1 \right]$ and $q_i = \Pr \left[ I_i=0 \right], \hspace{2mm} i=1,\ldots,n+1$. Obviously, $\sum_{j=1}^{n}I_{j}$ and $I_{n+1}$ are mutually exclusive.\ The
participants decide to set up an RS scheme $\left(  \boldsymbol{\pi
},\boldsymbol{P}\right)  $, with the relative compensation vector
$\boldsymbol{P}$\ given by
\begin{equation}
P_{i}=\frac{f_{i}\times I_{i}}{\sum_{j=1}^{n+1}f_{j}\times I_{j}},\qquad
i=1,2,\ldots,n+1, \label{C5}%
\end{equation} 
where $f_{i}$, for $i=1,2,\ldots,n+1$, are strictly positive real
numbers.\ \newline The RS scheme $\left(  \boldsymbol{\pi},\boldsymbol{P}%
\right)  $\ defined via the relative compensations (\ref{C5}) is a special
case of the RS scheme considered in Example 1.\ Such RS schemes with an active
administrator have been investigated in detail in Dhaene and Milevsky (2024),
who investigate fair methods for the surviving participants to share the total
investment among themselves if one or more survive.\ In case all participants
pass away, the administrator receives the full proceeds of the fund. An RS
scheme of this type is often called a tontine fund.\ Dhaene and Milevsky
(2024) call $f_{i}$ the number of \textit{tontine shares} invested in the
tontine fund, and describe the RS scheme defined by (\ref{C5}) as a scheme
where the proceeds of the fund are equally shared among all surviving tontine
shares, where the total number of surviving tontine shares is given by the
denominator in (\ref{C5}).\ They consider the situation where initial
investments (wealth) and survival probabilities (health) vary among
participants, which is called the heterogeneous case.\ As a special case, they
also examine the situation where all participants invest the same amount and
the random variables $I_{1},I_{2},\ldots,I_{n}$ are i.i.d., which they refer to as the
homogeneous case.\ Denuit and Robert (2025) consider a similar scheme (with a
passive administrator, see further), where the $f_{i}$ are related to (what
they call) \textit{protection units} from the investment fund, and where the
total value of the tontine fund is divided equally among all claiming
units.\ An essential difference between the two approaches is that
Dhaene and Milevsky (2024) introduce the active administrator who contributes
to the investments and will own the proceeds of the fund in case not any
person survives, whereas Denuit and Robert (2025) assume a passive
administrator who does not contribute to the investments, and in case nobody
experiences the event under consideration, participants receive their initial
investment back.\ We will come back to this essential difference between the
two approaches in a further section of this paper.\ \newline Special cases of
the RS scheme $\left(  \boldsymbol{\pi},\boldsymbol{P}\right)  $\ defined via
the relative compensations (\ref{C5}) have been considered in several papers.\ Dhaene and
Milevsky (2024) consider the following choices for the claiming units $f_{i}$:%

\begin{equation}
f_{i}^{\text{DM}}=\frac{\pi_{i}}{p_i},\qquad
i=1,2,\ldots,n+1.
\end{equation}
Tavin (2023) proposes the following set of $f_{i}$:
\begin{equation}
f_{i}^{\text{T}}=\pi_{i},\qquad i=1,2,\ldots,n+1. \label{C91}%
\end{equation}
Denuit and Robert (2023) also consider the case of uniform $f_{i}$:
\begin{equation}
f_{i}^{\text{DR}}=1,\qquad i=1,2,\ldots,n+1.
\end{equation}
Finally, in Dhaene and Milevsky (2024), the RS scheme with
\[
f_{i}=\frac{1}{p_i},\qquad i=1,2,\ldots,n+1
\]
is also considered. Interpretations and motivations for any of these choices for the number of
claiming units $f_{i}$ can be found in the above-mentioned papers.\ \newline
We remark that the compensation vector $\boldsymbol{W}$ with relative
compensation vector determined by (\ref{C5}) can be rewritten as
\begin{equation}
W_{i}=\left(  \frac{\sum_{k1}^{n+1}\pi_{k}}{\sum_{j=1}^{n+1}f_{j}\times I_{j}%
}\right)  \times f_{i}\times I_{i},\qquad i=1,2,\ldots,n+1.
\end{equation}
This means that under the RS scheme $\left(  \boldsymbol{\pi},\boldsymbol{P}%
\right)  $ defined by (\ref{C5}), each surviving participant $i$ is
compensated the same time $1$ observable proportion of his surviving shares or
protection units $f_{i}$. The proportion is the random payment
per claiming protection unit, which is determined such that the full
allocation condition is fulfilled.\ $\hfill\vartriangleleft$
\end{example}

\begin{example}
Consider the loss vector $\left(  X_{1},X_{2},\ldots,X_{n}\right)  $,
describing the non-negative losses of the $n$ participants in the observation
period $\left[  0,1\right]  $.\ The $i$-th\ order statistic\ $X_{(i)}$\ of
$\left(  X_{1},X_{2},\ldots,X_{n}\right)  $ is the $i$-th smallest value in
$\left(  X_{1},X_{2},\ldots,X_{n}\right)  $. Hence,
\[
X_{(1)}\leq X_{(2)}\leq\cdots\leq X_{(n)}.
\]
Furthermore, we define $X_{(n+1)}$ as follows:
\begin{equation}
X_{(n+1)}=1\left(  \sum_{j=1}^{n}X_{j}=0\right)  . \label{C9a}%
\end{equation}
Notice that we only consider the order statistics for the losses of the $n$
participants.\ The random variable $X_{(n+1)}$ defined above is not an order
statistic. The notation $X_{(n+1)}$ is introduced only to make notations
uniform and simple.\ \newline Suppose that the participants set up a fund in
which each participant $i$ invests an amount $\pi_{i}$.\ Moreover, they appoint an
administrator, who contributes the amount $\pi_{n+1}$ to the fund.\ The $n+1$
agents determine the compensations according to the risk-sharing scheme
$\left(  \boldsymbol{\pi},\boldsymbol{P}\right)  $, with relative compensation
vector $\boldsymbol{P}$ determined by%
\begin{equation}
P_{i}=\frac{X_{(i)}}{\sum_{j=1}^{n+1}X_{(j)}},\qquad i=1,2,\ldots,n+1.
\end{equation}
The interpretation of this compensation-based RS scheme $\left(
\boldsymbol{\pi},\boldsymbol{P}\right)  $ is as follows.\ For participants
who are ordered in decreasing risk-bearing capacity (e.g., decreasing wealth
or decreasing age), a lower risk-bearing capacity leads to a higher
compensation.\ In case all participants have a zero-claim, the proceeds of the
fund are fully transferred to the administrator.\ In settings with a small number of participants who know one another well and are willing to contribute more than their `fair' share to support others in the pool, such a scheme can be particularly useful. Such a situation arises frequently in practice: within families, close-knit communities or groups bound by friendship or solidarity, members often accept cross-subsidization as a natural consequence of mutual support. For instance, healthier or wealthier members of a family may bear disproportionate share of total expenses without expecting strict actuarial fairness in return. $\hfill\vartriangleleft$
\end{example}

\begin{example}
Consider the vector $\left(  X_{1},X_{2},\ldots,X_{n}\right)  $, describing
the non-negative losses of the $n$ participants in the observation period
$\left[  0,1\right]  $. The participants set up a fund in which each participant $i$
invests an amount $\pi_{i}$.\ They also appoint an administrator, who
contributes the amount $\pi_{n+1}$ to the fund.\ The participants and the administrator decide to
share the proceeds of the fund according to the RS scheme $\left(
\boldsymbol{\pi},\boldsymbol{P}\right)  $, with the relative compensation
vector $\boldsymbol{P}$ being $\left(  X_{1},X_{2},\ldots,X_{n}\right)  $ -
measurable. This means that the randomness of $\boldsymbol{P}$ is only due to
the randomness of the vector $\left(  X_{1},X_{2},\ldots,X_{n}\right)  $%
.$\ $Hence, there exist functions $g_{i}:$\ $\left(
\mathbb{R}
^{+}\right)  ^{n}\rightarrow%
\mathbb{R}
^{+}$ such that
\begin{equation}
P_{i}=g_{i}\left(  X_{1},X_{2},\ldots,X_{n}\right)  ,\qquad i=1,\ldots,n+1.
\label{C90}%
\end{equation}
Notice that the relative compensation vector will in general also depend on
the initial investment vector $\boldsymbol{\pi}$ of the RS scheme $\left(
\boldsymbol{\pi},\boldsymbol{P}\right)  $ under consideration.\ But as
$\boldsymbol{\pi}$ is fixed, we do not explicitly indicate the dependence of
$\boldsymbol{\pi}$ in the notation of the relative compensation vector. The
full allocation condition (\ref{C11a}) \ can now be expressed
as follows:
\[
\sum_{i=1}^{n+1}g_{j}\left(  X_{1},X_{2},\ldots,X_{n}\right)  =1.
\]
Moreover, the relative compensation vector is assumed to satisfy
\begin{equation}
g_{n+1}\left(  X_{1},X_{2},\ldots,X_{n}\right)  =1\left(  \sum_{j=1}^{n}%
g_{j}\left(  X_{1},X_{2},\ldots,X_{n}\right)  =0\right)  , \label{C90b}%
\end{equation}
which implies that $\sum_{i=1}^{n}g_{i}\left(  X_{1},X_{2},\ldots,X_{n}\right)
$ and $g_{n+1}\left(  X_{1},X_{2},\ldots,X_{n}\right)  $ are mutually
exclusive. \ 

The RS scheme considered in Example 1 is a special case of (\ref{C90}), with the relative compensation vector determined from
\[
P_{i}=g_{i}\left(  X_{1},X_{2},\ldots,X_{n}\right)  =\frac{X_{i}}{\sum
_{j=1}^{n+1}X_{j}},\qquad i=1,\ldots,n+1,
\]
with $X_{n+1}$ given by (\ref{C35b}).\ \newline Also the RS scheme considered
in Example 2 is a special case, with the relative compensations defined by
\[
P_{i}=g_{i}\left(  X_{1},X_{2},\ldots,X_{n}\right)  =\frac{f_{i}\times I_{i}%
}{\sum_{j=1}^{n+1}f_{j}\times I_{j}},\qquad i=1,\ldots,n+1,
\]
with the Bernoulli random variables $I_{i}$ as defined in (\ref{C5a}) and (\ref{C5b}%
).\ \newline Another special case of (\ref{C90}) arises by making the
following choice for the relative compensations:
\[
P_{i}=g_{i}\left(  X_{1},X_{2},\ldots,X_{n}\right)  =\frac{X_{(i)}}{\sum
_{j=1}^{n+1}X_{(j)}},\qquad i=1,\ldots,n+1,
\]
where $X_{(n+1)}$ is defined in (\ref{C9a}). This special case was considered
in Example 3.\ $\hfill\vartriangleleft$
\end{example}

\begin{example} \label{example:continuous}
Returning to the risk-sharing scheme in Example 1, we investigate a mixed distribution for $X_i$. For $i = 1, 2, \ldots, n$, let
\begin{equation}
    X_i = I_i Y_i, \label{eq:Xi-def}
\end{equation}
where $I_i \sim \mathrm{Bernoulli}(p_i)$ and $Y_i \sim \mathrm{Gamma}(\alpha_i, \theta)$, where the shape $\alpha_i > 0$ depends on the participant and the scale $\theta > 0$ is common across participants. We assume that $(I_1, \ldots, I_n)$ has mutually independent components, $(Y_1, \ldots, Y_n)$ has mutually independent components, and the $\sigma$-algebras $\sigma(I_1, \ldots, I_n)$ and $\sigma(Y_1, \ldots, Y_n)$ are independent. Let
\begin{equation}
X_{n+1} = 1\Big(\sum_{j=1}^n X_j = 0\Big) = 1\Big(\sum_{j=1}^n I_j = 0\Big), 
\label{eq:Xn+1-def}
\end{equation}
where the second equality holds because $Y_j > 0$ almost surely. Each loss has a strictly positive point mass at zero, $\Pr[X_i = 0] = q_i = 1-p_i$. 
Suppose $P_i$ for $i=1,2,\ldots,n+1$ is defined as in (\ref{C35}). Since $\sum_{j=1}^n I_j Y_j = 0$ if and only if $\sum_{j=1}^n I_j = 0$, the denominator is almost surely positive, and for $i = 1,2, \ldots,n$,
\begin{equation}
    P_i = \begin{cases} \dfrac{Y_i}{\sum_{j=1}^n I_j Y_j} & \text{if } I_i = 1,\\[6pt] 0 & \text{if } I_i = 0 \end{cases} 
\label{eq:Pi-cases}
\end{equation}
and
\begin{equation}
    P_{n+1} = 1\Big(\sum_{j=1}^n I_j = 0\Big).
\end{equation}
For mutually independent $(Y_1,\ldots, Y_n)$ defined above, we have the following results:
\begin{enumerate}
    \item[(a)] $T = \sum_{j=1}^n Y_j \sim \mathrm{Gamma}(\sum_{j=1}^n \alpha_j, \theta)$;
    \item[(b)] $T$ is independent of $(Y_1/T, \ldots, Y_n/T)$;
    \item[(c)] $(Y_1/T, \ldots, Y_n/T) \sim \mathrm{Dirichlet}(\alpha_1, \ldots, \alpha_m)$.
    \item[(d)] $ Y_i/T\sim \text{Beta}\left( \alpha_i, \sum_{j\neq i} \alpha_j \right)$
\end{enumerate}
These results are proved in Theorem 4.1 and Theorem 4.2 of Devroye (1986) on page 594 and page 595 (respectively). Define the claimant set $S := \{j \in \{1,\ldots,n \} \mid I_j = 1\}$. The above results along with the $\sigma$-algebras $\sigma(I_1,\ldots,I_n)$ and $\sigma(Y_1,\ldots,Y_n)$ being independent show that for a participant $k$ in the claimant set $S$,
\begin{equation}
P_k = R_k, \quad R_k \sim \mathrm{Beta}\left(\alpha_k, \sum_{j \in S\setminus \{k\}}\alpha_j \right), \qquad \text{independent of }  \sum_{j \in S} Y_j \sim \mathrm{Gamma}\left(\sum_{j\in S} \alpha_j, \theta\right). \label{eq:cond-rep}
\end{equation}
\end{example}

\section{Actuarial fairness of compensation-based risk-sharing with an active
administrator}

\subsection{Actuarially fair risk-sharing schemes}

An RS scheme $\left(  \boldsymbol{\pi},\boldsymbol{P}\right)  $ is said to be
actuarially fair for the participants if the time $1$ value of each
participant's initial investment $\pi_{i}$ in the fund is equal to the
expected value of the compensation $W_{i}$ that he will receive at time $1$.
This means that no participant experiences a gain or loss on average by
joining the pool. Actuarial fairness of particular compensation-based RS
schemes has been investigated in Bernard et al.\ (2024), Milevsky and Dhaene
(2024) and Denuit and Robert (2025), amongst others.\ 

\begin{definition}
The RS scheme $\left(  \boldsymbol{\pi},\boldsymbol{P}\right)  $ with an
active administrator is actuarially fair for each participant if the following
conditions hold:
\[
\pi_{i}=E\left[  W_{i}\right]  ,\qquad i=1,2,\ldots,n.
\]

\end{definition}

Taking into account (\ref{C1aa}), we can rewrite the actuarial fairness
conditions for the $n$ participants as
\begin{equation}
\pi_{i}=\left(  \sum_{j=1}^{n+1}\pi_{j}\right)  \times E\left[  P_{i}\right]
,\qquad i=1,2,\ldots,n. \label{C14}%
\end{equation}

Notice that in a real world context, the particular RS scheme that is chosen
by a group of participants may depend on the social cohesion between them,
ranging from solidarity to pure individualism.\ Especially in small pools of
connected participants (e.g., family members or tribe members), actuarial fairness may not be
the first concern and may be replaced by a form of organized transfer, e.g.,
from the elder participants to the younger, or from the richer to the poorer
ones.\ In this section however, we will further investigate actuarial fairness
of RS schemes.\ 

In our general compensation-based RS set-up with an active administrator,
all proceeds of the fund are transferred to the administrator in case no
participant receives a positive compensation.\ In return, the administrator
pays an initial investment $\pi_{n+1}$ for this benefit.\ Let us now define
actuarial fairness for the administrator.

\begin{definition}
The RS scheme $\left(  \boldsymbol{\pi},\boldsymbol{P}\right)  $ with an active administrator is
actuarially fair for the active administrator if
\[
\pi_{n+1}=E\left[  W_{n+1}\right]  .
\]

\end{definition}

Taking into account (\ref{C31}) and (\ref{C1aa}), we can rewrite this
actuarial fairness condition for the administrator as follows:
\begin{equation}
\pi_{n+1}=\left(  \sum_{j=1}^{n+1}\pi_{j}\right)  \times\Pr\left[
P_{n+1}=1\right]  . \label{C15}%
\end{equation}
As $\Pr\left[  P_{n+1}=0\right]  $ is strictly positive by assumption, see
(\ref{C30}), the actuarial fairness condition for the administrator can also
be expressed as follows:
\begin{equation}
\pi_{n+1}=\left(  \sum_{j=1}^{n}\pi_{j}\right)  \times\frac{\Pr\left[
P_{n+1}=1\right]  }{\Pr\left[  P_{n+1}=0\right]  }. \label{C17}%
\end{equation}

We can summarize relations (\ref{C15}) and (\ref{C17}) in the following lemma.

\begin{lemma}
    If the RS scheme $\left(  \boldsymbol{\pi},\boldsymbol{P}\right)  $ is actuarially fair for the active administrator, then we have that
    \begin{equation}
        \sum_{j=1}^{n+1} \pi_j = \frac{1}{\Pr[P_{n+1}=0]} \times \sum_{j=1}^n \pi_j = \frac{1}{\Pr [P_{n+1}=1]} \times \pi_{n+1}
    \end{equation}
\end{lemma}

In case the relative compensations $P_{i}$ of the participants are i.i.d., we have
that
\[
\Pr\left[  P_{n+1}=1\right]  =\Pr\left[  P_{1}=P_{2}=\cdots=P_{n}=0\right]
=\left(  \Pr\left[  P_{1}=0\right]  \right)  ^{n}.
\]
This means that if the number of participants $n$ is sufficiently large, we
find that $\Pr\left[  P_{n+1}=1\right]  \approx0$ and $\Pr\left[
P_{n+1}=0\right]  \approx1$, which implies that the actuarially fair initial
investment $\pi_{n+1}$ of the administrator is close to zero.$\ $

Let us now consider the relation between actuarial fairness for the
participants and actuarial fairness for the active administrator.

\begin{proposition}
If the RS scheme $\left(  \boldsymbol{\pi},\boldsymbol{P}\right)  $ is
actuarially fair for the $n$ participants, then it is also actuarially fair
for the active administrator.
\label{Prop1}
\end{proposition}

\begin{proof}
From the actuarial fairness conditions (\ref{C14}) for the $n$ participants,
we find that
\[
\sum_{k=1}^{n}\pi_{k}=\left(  \sum_{j=1}^{n+1}\pi_{j}\right)  \times\sum
_{k=1}^{n}E\left[  P_{k}\right]  .
\]
Taking into account (\ref{C34a}) leads to the actuarial fairness condition
(\ref{C17}) for the administrator.
\end{proof}

From the previous proposition, we can conclude that if the RS scheme $\left(
\boldsymbol{\pi},\boldsymbol{P}\right)  $ is not actuarially fair for the active
administrator, then it can also not be actuarially fair for all
participants.\ In particular, this situation will occur in case the
administrator makes an investment of zero.\ This observation is further
explored in the following proposition.\ 

\begin{remark}
From Proposition 1, we can conclude that if the active administrator of the RS scheme $\left(  \boldsymbol{\pi},\boldsymbol{P}\right)  $ makes a zero initial investment, i.e., $\pi_{n+1}=0$, then there must be at least one participant $i$ whose investment $\pi_i$ exceeds his expected compensation, i.e., $\pi_{i}>E\left[  W_{i}\right]  $. In other words, such an RS scheme cannot be actuarially fair.
\end{remark}

In the following proposition, we consider several necessary and sufficient
conditions for actuarial fairness for the $n+1$ agents involved in the RS
scheme $\left(  \boldsymbol{\pi},\boldsymbol{P}\right)  $.\ 

\begin{proposition}
The RS scheme $\left(  \boldsymbol{\pi},\boldsymbol{P}\right)  $ with active
administrator is actuarially fair for the $n+1$ agents if and only if any of the
following conditions is satisfied:\newline\textbf{Condition 1}: The RS scheme
$\left(  \boldsymbol{\pi},\boldsymbol{P}\right)  $ satisfies
\begin{equation}
\pi_{i}=\left(  \sum_{j=1}^{n+1}\pi_{j}\right)  \times E\left[  P_{i}\right]
,\qquad i=1,2,\ldots,n+1. \label{C15c}%
\end{equation}
\newline\textbf{Condition 2}: The RS scheme $\left(  \boldsymbol{\pi
},\boldsymbol{P}\right)  $ satisfies%
\begin{equation} 
\pi_{i}=\left(  \sum_{j=1}^{n}\pi_{j}\right)  \times\frac{E\left[
P_{i}\right]  }{\Pr\left[  P_{n+1}=0\right]  },\qquad i=1,2,\ldots,n+1.
\label{C15b}%
\end{equation}
\newline\textbf{Condition 3}: The RS scheme $\left(  \boldsymbol{\pi
},\boldsymbol{P}\right)  $ satisfies
\begin{equation}
\pi_{i}=\pi_{n+1}\times\frac{E\left[  P_{i}\right]  }{\Pr\left[
P_{n+1}=1\right]  },\qquad i=1,2,\ldots,n+1. \label{C15a}%
\end{equation}

\end{proposition}

\begin{proof}
In order to prove the implications
$$ \text{Condition 1} \Rightarrow \text{Condition 2} \Rightarrow \text{Condition 3} \Rightarrow \text{Condition 1}, $$
it suffices to prove that each of the three conditions implies actuarial fairness for the active administrator, see Lemma 1.
\newline Obviously, Condition 1 implies actuarial fairness for the active administrator.
From Condition 2, for agent $n+1$, we find that
$$ \pi_{n+1} = \left( \sum_{j=1}^n \pi_j \right) \times \frac{\Pr [P_{n+1}=1]}{\Pr [P_{n+1} = 0]}, $$
which is equivalent with actuarial fairness for the active administrator.
\newline Summing over all agents, Condition 3 leads to:
$$ \sum_{j=1}^{n+1} \pi_i = \pi_{n+1} \times \frac{1}{\Pr [P_{n+1} = 1]}, $$
which is again equivalent with actuarial fairness for the active administrator.
\newline This ends the proof. 
\end{proof}

\subsection{Risk-sharing rules and actuarial fairness}

So far, we have considered the actuarial fairness of a given \textit{RS scheme}
$\left(  \boldsymbol{\pi},\boldsymbol{P}\right)  $ with active administrator.
Let us now consider an \textit{RS rule} $\boldsymbol{P}$ with an active
administrator, which transforms any investment vector $\boldsymbol{\pi}$ into
a relative compensation vector $\boldsymbol{P}\left[  \boldsymbol{\pi}\right]
$:%
\begin{equation}
\boldsymbol{\pi}\rightarrow\boldsymbol{P}\left[  \boldsymbol{\pi}\right]
\text{.}%
\end{equation}
This means that any investment vector $\boldsymbol{\pi}$\ leads to the
compensation vector $\boldsymbol{W}\left[  \boldsymbol{\pi}\right]  $, with%
\begin{equation}
\boldsymbol{W}\left[  \boldsymbol{\pi}\right]  =\left(  \sum_{j=1}^{n+1}%
\pi_{j}\right)  \times\boldsymbol{P}\left[  \boldsymbol{\pi}\right]  \text{,}
\label{C61}%
\end{equation}
see Definition 3. Let us now consider a special type of RS rules
$\boldsymbol{P}$, which satisfy the following indifference property:
\begin{equation}
\boldsymbol{P}\left[  c\times\boldsymbol{\pi}\right]  =\boldsymbol{P}\left[
\boldsymbol{\pi}\right]  \text{,}\qquad\text{\ for any }c>0\text{ and any
investment vector }\boldsymbol{\pi}\text{.} \label{C61a}%
\end{equation}
This means that the RS rule $\boldsymbol{P}$ is such that if all participants
increase their initial investment by a constant proportion, e.g., by 20\%, then
their relative compensation vector remains unchanged, which seems to be a
reasonable property.\ Obviously any RS rule of which the relative compensation vector $\boldsymbol{P}$ is independent of the initial investment vector satisfies the indifference property (\ref{C61a}). In the following proposition, we consider contribution
vectors of an RS rule $\boldsymbol{P}$ that satisfies the indifference property
(\ref{C61a}).

\begin{proposition}
Consider the RS rule $\boldsymbol{P}$ with active administrator that
satisfies the indifference property (\ref{C61a}). Then for any positive $c$
and any investment vector $\boldsymbol{\pi}$, one has that%
\[
\mathbf{W}\left[  c\times\boldsymbol{\pi}\right]  =c\times\mathbf{W}\left[
\boldsymbol{\pi}\right]
\]

\end{proposition}

\begin{proof}
For any investment vector $\boldsymbol{\pi}$, any $c>0$ and any participant $i$, we
find that
\begin{align*}
W_{i}\left[  c\times\boldsymbol{\pi}\right]   &  =c\times\left(  \sum
_{j=1}^{n+1}\pi_{j}\right)  \times P_{i}\left[  c\times\boldsymbol{\pi}\right]
\\
&  =c\times\left(  \sum_{j=1}^{n+1}\pi_{j}\right)  \times P_{i}\left[
\boldsymbol{\pi}\right] \\
&  =c\times W_{i}\left[  \boldsymbol{\pi}\right]  .
\end{align*}
This ends the proof.\ \ 
\end{proof}

Suppose now that a group of participants decides to use the RS rule $\boldsymbol{P}$
with an active administrator, which satisfies the indifference property
(\ref{C61a}).\ The next question they have to answer is then what particular
initial investment vector $\boldsymbol{\pi}$ to choose.\ It may be reasonable
to select an investment vector $\boldsymbol{\pi}$ such that the RS scheme
$\left(  \boldsymbol{\pi},\boldsymbol{P}\left[  \boldsymbol{\pi}\right]
\right)  $ is actuarially fair. In the following proposition, we show that
this problem has in general no unique solution.\ 

\begin{proposition}
Consider the RS rule $\boldsymbol{P}$ with an active administrator satisfying
the indifference property (\ref{C61a}).\ Suppose that the RS scheme $\left(
\boldsymbol{\pi}^{\ast},\boldsymbol{P}\left[  \boldsymbol{\pi}^{\ast}\right]
\right)  \boldsymbol{\ }$is actuarially fair for all participants, then for any
$c>0$, also the RS scheme $\left(  c\times\boldsymbol{\pi}^{\ast
},\boldsymbol{P}\left[  c\times\boldsymbol{\pi}^{\ast}\right]  \right)  $ is
actuarially fair for all participants.\ 
\end{proposition}

\begin{proof}
From Proposition\ 3, we find for any participant $i$ that
\[
E\left[  W_{i}\left[  c\times\boldsymbol{\pi}^{\ast}\right]  \right]  =c\times
E\left[  W_{i}\left[  \boldsymbol{\pi}^{\ast}\right]  \right]  .
\]
The actuarial fairness of $\left(  \boldsymbol{\pi}^{\ast},\boldsymbol{P}%
\left[  \boldsymbol{\pi}^{\ast}\right]  \right)  $ can be expressed as
\[
E\left[  W_{i}\left[  \boldsymbol{\pi}^{\ast}\right]  \right]  =\pi_{i}^{\ast
},\qquad i=1,2,\ldots,n+1.
\]
Combining both expressions leads to
\[
E\left[  W_{i}\left[  c\times\boldsymbol{\pi}^{\ast}\right]  \right]
=c\times\pi_{i}^{\ast},\qquad i=1,2,\ldots,n+1,
\]
which are the actuarial fairness conditions for $\left(  c\times
\boldsymbol{\pi}^{\ast},\boldsymbol{P}\left[  c\times\boldsymbol{\pi}^{\ast
}\right]  \right)  $.
\end{proof}

We can conclude that for any RS rule $\boldsymbol{P}$ satisfying the
indifference property (\ref{C61a}), one has that if the RS scheme $\left(
\boldsymbol{\pi}^{\boldsymbol{\ast}},\boldsymbol{P}\left[  \boldsymbol{\pi
}^{\ast}\right]  \right)  $ is actuarially fair for all participants, then also
the RS scheme $\left(  c\times\boldsymbol{\pi}^{\ast},\boldsymbol{P}\left[
c\times\boldsymbol{\pi}^{\ast}\right]  \right)  $ is actuarially fair for all
participants.\ This means that for such RS rules, actuarially fair investment
vectors $\boldsymbol{\pi}$ are only defined up to a positive constant
factor.\ For RS rules satisfying the independence property (\ref{C61a}),
Proposition 2 may be a guide for choosing the appropriate actuarially fair set
of initial investments: One can either first determine the level of the
aggregate investments of the $n+1$ agents, or the level of the aggregate
investment of the $n$ participants, or the level of the individual investment
of the administrator, and then determine the individual actuarial fair
investments $\pi_{i}$ by the corresponding Conditions 1, 2 or 3.

\begin{example}
In order to illustrate the previous propositions, consider a group of $n$
participants who are exposed to the losses $X_{1},$ $X_{2},\ldots,X_{n}$,
respectively.\ Suppose they agree to use the RS rule $\boldsymbol{P}$, where
for any RS scheme $\left(  \boldsymbol{\pi},\boldsymbol{P}\left[
\boldsymbol{\pi}\right]  \right)  $, the relative compensation vector $\boldsymbol{P}%
\left[  \boldsymbol{\pi}\right]  $ is given by (\ref{C35}), as considered in
Example 1:
\[
P_{i}\left[  \boldsymbol{\pi}\right]  =\frac{X_{i}}{\sum_{k=1}^{n+1}X_{k}%
},\qquad i=1,2,\ldots,n+1,
\]
with $X_{n+1}$ defined by (\ref{C35b}). As the random losses $X_{1},$ $X_{2}%
,\ldots,X_{n}$ are assumed to be independent of $\boldsymbol{\pi}$, we have that the relative compensation vectors of
the RS rule $\boldsymbol{P}$\ satisfy the indifference property (\ref{C61a}%
).\ 
\newline Let us assume that the participants decide to choose an investment
vector $\boldsymbol{\pi}\ $such that the RS scheme $\left(  \boldsymbol{\pi
},\boldsymbol{P}\left[  \boldsymbol{\pi}\right]  \right)  $\ is actuarially
fair for any of them.\ This means that the investment vector follows from the
set of equations (\ref{C15c}):%
\begin{equation}
\pi_{i}=\left(  \sum_{j=1}^{n+1}\pi_{j}\right)  \times E\left[  \frac{X_{i}%
}{\sum_{k=1}^{n+1}X_{k}}\right]  ,\qquad i=1,2,\ldots,n+1. \label{C70}%
\end{equation}
In this case an actuarially fair investment vector is only defined up to a
positive constant factor, that is, if $\boldsymbol{\pi}^{\ast}$ is a solution of
(\ref{C70}), then also $c\times\boldsymbol{\pi}^{\ast}$ is a solution, for any
$c>0$. Applying this RS rule in practice, we could first determine a reference
solution $\boldsymbol{\pi}$ of (\ref{C70}), and in a second step, determine
the investment vector $c\times\boldsymbol{\pi}$ with $c$ such that
$\Pr\left[  c\times\sum_{j=1}^{n+1}\pi_{j}>\sum_{j=1}^{n+1}X_{k}\right]  $
is sufficiently large.\ As an extreme case, suppose that $c$ is chosen
such that $\Pr\left[  c\times\sum_{j=1}^{n+1}\pi_{j}>\sum_{j=1}^{n+1}%
X_{k}\right]  =1$, then we have that
\[
W_{i}=\left(  c\times\sum_{j=1}^{n+1}\pi_{j}\right)  \times\frac{X_{i}%
}{\sum_{k=1}^{n+1}X_{k}}\geq X_{i},\qquad i=1,2,\ldots,n+1,
\]
which means that each compensation $W_{i}$ is always larger than its
corresponding claim $X_{i}$ in this case.\ \hfill$\lhd$
\end{example}

\begin{example}
Consider the setting of Example 2 with the number of protection units $f_{i}%
$\ chosen to be equal to the initial investment $\pi_{i}$. This means that we
consider the RS rule $\boldsymbol{P}$, where for any RS scheme $\left(
\boldsymbol{\pi},\boldsymbol{P}\left[  \boldsymbol{\pi}\right]  \right)  $,
each participant and the administrator receive a level of compensation
proportional to their initial contribution in case the predefined event
occurs:
\begin{equation}
W_{i}\left[  \boldsymbol{\pi}\right]  =\left(  \sum_{j=1}^{n+1}\pi_{j}\right)
\times\frac{\pi_{i}\times I_{i}}{\sum_{j=1}^{n+1}\pi_{j}\times I_{j}},\qquad
i=1,2,\ldots,n+1. \label{C36}%
\end{equation}
In these equations, the $I_i$ are Bernouilli-distributed random variables with $ \Pr \left[ I_i=1 \right]=p_i$ and $\Pr \left[ I_i=0 \right]=q_i]$. 

We call this rule that was proposed in Tavin (2023) the Tavin RS rule, see
(\ref{C91}). Obviously, the RS rule $\boldsymbol{P}$ satisfies the
indifference property (\ref{C61a}), implying that actuarially fair investments
vectors are only determined up to a positive constant factor.\ From
(\ref{C36}) it follows that the actuarial fairness conditions for all participants
in this RS scheme $\left(  \boldsymbol{\pi},\boldsymbol{P}\left[
\boldsymbol{\pi}\right]  \right)  $ can be written as follows:%
\[
1=\left(  \sum_{j=1}^{n+1}\pi_{j}\right)  \times E\left[  \frac{I_{i}}
{\sum_{j=1}^{n+1}\pi_{j}\times I_{j}}\right]  ,\qquad i=1,2,\ldots,n+1.
\]
Obviously, for actuarially fair initial investments the expectations $E\left[  \frac{I_{i}%
}{\sum_{j=1}^{n+1}\pi_{j}\times I_{j}}\right]  \ $ have to be equal for all agents.\ It is important to note  the fact that an RS scheme which is not actuarially fair is not necessarily a
`wrong' choice.\ Suppose that two participants of different age each pay the
same initial investment.\ Then, upon survival, they will receive the same
compensation.\ However, the one with the higher survival probability (the
younger one, let's say) will more likely survive and hence, is favored.\ As
clearly discussed in Tavin (2023), this RS scheme accommodates a reallocation
of wealth that is favorable to those who are likely to survive longer.\ Such
a reallocation can be a valuable objective, for example, in a given small community of
family members.
\end{example}

\begin{example}
Consider the Tavin RS rule $\boldsymbol{P}$ based on
the RS schemes $\left(  \boldsymbol{\pi},\boldsymbol{P}\left[  \boldsymbol{\pi
}\right]  \right)  $ of Example 7, with two participants and an active administrator.

From (\ref{C36}) it follows that the actuarial fairness conditions for all participants can be written as follows:%
\begin{equation*}
    1 = (\pi_1+\pi_2+\pi_3) \times \mathbb{E} \left[ \frac{I_i}{\pi_1 I_1+\pi_2 I_2+\pi_3 I_3} \right], \quad i = 1, 2, 3.
\end{equation*}
Let us assume that $I_1$ and $I_2$ are mutually independent. In this setting, we find that
\begin{equation*}
     \mathbb{E} \left[ \frac{I_1}{\pi_1 I_1+\pi_2 I_2+\pi_3 I_3} \right] = \left( \frac{p_1 p_2}{\pi_1 + \pi_2} + \frac{p_1 q_2}{\pi_1} \right),
\end{equation*}
while
\begin{equation*}
     \mathbb{E} \left[ \frac{I_2}{\pi_1 I_1+\pi_2 I_2+\pi_3 I_3} \right] = \left( \frac{p_1 p_2}{\pi_1 + \pi_2} + \frac{p_2 q_1}{\pi_2} \right)
\end{equation*}
and 
\begin{equation*}
     \mathbb{E} \left[ \frac{I_3}{\pi_1 I_1+\pi_2 I_2+\pi_3 I_3} \right] =  \frac{q_1 q_2}{\pi_3}.
\end{equation*}
This results in the following system of equations which characterizes the set of all initial investments $\left( \pi_1, \pi_2, \pi_3 \right)$ which leads to an actuarially fair Tavin RS scheme (\ref{C36})

\begin{empheq}[left=\empheqlbrace]{align}
  1 &= (\pi_1+\pi_2+\pi_3) \times \left( \frac{p_1 p_2}{\pi_1 + \pi_2} + \frac{p_1 q_2}{\pi_1} \right)  \label{First_condition_TVRS_example}\\
  1 &= (\pi_1+\pi_2+\pi_3) \times \left( \frac{p_1 p_2}{\pi_1 + \pi_2} + \frac{p_2 q_1}{\pi_2} \right) \label{Second_condition_TVRS_example}\\
  1 &= (\pi_1+\pi_2+\pi_3) \times \left( \frac{q_1 q_2}{\pi_3} \right). \label{Third_condition_TVRS_example}
\end{empheq}
From (\ref{First_condition_TVRS_example}) and (\ref{Second_condition_TVRS_example}), we find that

\begin{equation} \label{First_derivation_TVRS_example}
    \pi_2 = \frac{p_2 q_1}{p_1 q_2} \pi_1,
\end{equation}
while equation (\ref{Third_condition_TVRS_example}) leads to

\begin{equation} \label{Second_derivation_TVRS_example}
    \pi_1 + \pi_2 = \pi_3 \times \frac{1-q_1q_2}{q_1q_2}.
\end{equation}
Substituting (\ref{First_derivation_TVRS_example}) in (\ref{Second_derivation_TVRS_example}) gives rise to

\begin{equation} \label{First_result_TVRS_example}
    \pi_1 = \pi_3 \times \frac{p_1}{q_1} \times \frac{1-q_1q_2}{p_1q_2 + p_2q_1},
\end{equation}
and also

\begin{equation} \label{Second_result_TVRS_example}
    \pi_2 = \pi_3 \times \frac{p_2}{q_2} \times \frac{1-q_1q_2}{p_1q_2 + p_2q_1}.
\end{equation}

It is easy to verify that any $\left( \pi_1, \pi_2, \pi_3 \right)$ with $\pi_3 >0$ and $\pi_1$ and $\pi_2$ given by (\ref{First_result_TVRS_example}) and (\ref{Second_result_TVRS_example}), respectively, satisfies the system of equations (\ref{First_condition_TVRS_example}), (\ref{Second_condition_TVRS_example}) and (\ref{Third_condition_TVRS_example}). We conclude that the set of all initial investments $\left( \pi_1, \pi_2, \pi_3 \right)$ that lead to an actuarially fair Tavin RS scheme (\ref{C36}) with two participants is characterized by (\ref{First_result_TVRS_example}) and (\ref{Second_result_TVRS_example}), where $\pi_3$ can be any positive real number.
\end{example}

\begin{example}
    Consider a group of $2$ participants. Each of them may experience a loss $X_i$ in the observation period $\left[0,1\right]  $. Define the loss faced by participant 1 by
    \[
    X_1 =
    \begin{cases}
    0, & \text{ with probability } q_1,\\[0.2cm]
    1, &  \text{ with probability } d_1,\\[0.2cm]
    2, & \text{ with probability } p_1=1-q_1-d_1.
    \end{cases}
    \]
    Similarly, define the loss faced by participant two by:
    \[
    X_2 =
    \begin{cases}
    0, & \text{ with probability } q_2,\\[0.2cm]
    1, & \text{ with probability } p_2=1-q_2.
    \end{cases}
    \]
    \noindent Then define the random variable for the active administrator:
    \[
    X_{3} = \mathbf{1}\left( \sum_{j=1}^{2} X_j = 0 \right).
    \]
    Consider the Tavin RS rule $\boldsymbol{P}$ based on the RS schemes $\left(  \boldsymbol{\pi},\boldsymbol{P}\left[  \boldsymbol{\pi}\right]  \right)  $, with two participants and an active administrator. From Example 8, it follows that the actuarial fairness conditions for all participants can be written as follows:
    \begin{equation*}
    1 = (\pi_1+\pi_2+\pi_3) \times \mathbb{E} \left[ \frac{X_i}{\pi_1 X_1+\pi_2 X_2+\pi_3 X_3} \right], \quad i = 1, 2, 3.
    \end{equation*}
    Let us assume that $X_1$ and $X_2$ are mutually independent. In this setting, we find that
    \begin{equation*}
         \mathbb{E} \left[ \frac{X_1}{\pi_1 X_1+\pi_2 X_2+\pi_3 X_3} \right] = 
         \left(\frac{(p_1+d_1)q_2}{\pi_1} + \frac{d_1p_2}{\pi_1+\pi_2} + \frac{2p_1p_2}{2\pi_1+\pi_2} \right),
    \end{equation*}
    while
    \begin{equation*}
         \mathbb{E} \left[ \frac{X_2}{\pi_1 X_1+\pi_2 X_2+\pi_3 X_3} \right] = \left( \frac{p_2 q_1}{\pi_2} + \frac{p_2d_1}{\pi_1 + \pi_2} + \frac{p_2p_1}{2\pi_1 +\pi_2} \right)
    \end{equation*}
    and 
    \begin{equation*}
         \mathbb{E} \left[ \frac{X_3}{\pi_1 X_1+\pi_2 X_2+\pi_3 X_3} \right] =  \frac{q_1 q_2}{\pi_3}.
    \end{equation*}
    This results in the following system of equations which characterizes the set of all initial investments $\left( \pi_1, \pi_2, \pi_3 \right)$ which leads to an actuarially fair Tavin RS scheme. 
    \begin{empheq}[left=\empheqlbrace]{align}
      1 &= (\pi_1+\pi_2+\pi_3) \times  \left(\frac{(p_1+d_1)q_2}{\pi_1} + \frac{d_1p_2}{\pi_1+\pi_2} + \frac{2p_1p_2}{2\pi_1+\pi_2} \right)\\
      1 &= (\pi_1+\pi_2+\pi_3) \times \left( \frac{p_2 q_1}{\pi_2} + \frac{p_2d_1}{\pi_1 + \pi_2} + \frac{p_2p_1}{2\pi_1 +\pi_2} \right)\\
      1 &= (\pi_1+\pi_2+\pi_3) \times \left( \frac{q_1 q_2}{\pi_3} \right). 
      \label{}
    \end{empheq}
    Let
    \begin{equation}
        \Delta = \sqrt{\bigl[\,2(d_1+p_1)q_2 + p_1 p_2 - p_2 q_1\,\bigr]^{2} \,+\, 8\,p_2\, q_1\, (d_1+p_1)\, q_2}.
    \end{equation}    
    We find that
    \begin{align}
        \pi_2 &= \frac{\Delta \,-\, \bigl[\,2(d_1+p_1)q_2 + p_1 p_2 - p_2 q_1\,\bigr]}{2\,(d_1+p_1)\,q_2}\;\pi_1, \\
        \pi_1 &= \frac{\Delta \,+\, \bigl[\,2(d_1+p_1)q_2 + p_1 p_2 - p_2 q_1\,\bigr]}{4\,p_2\,q_1}\;\pi_2.
    \end{align}
    and so
    \begin{align} \label{First_derivation_TVRS_second_example}
        \pi_1 &= \frac{(1-q_1 q_2)\,\bigl[\,\Delta \,+\, p_2(p_1 - q_1)\,\bigr]}{2\,q_1 q_2\,\bigl(d_1 q_2 + p_1 q_2 + p_1 p_2 + p_2 q_1\bigr)}\;\pi_3,\\
    \end{align}
    \begin{align} \label{Second_derivation_TVRS_second_example}
        \pi_2 &= \frac{(1-q_1 q_2)\,\bigl[\,2(d_1+p_1)q_2 + p_1 p_2 + 3 p_2 q_1 \,-\, \Delta\,\bigr]}{2\,q_1 q_2\,\bigl(d_1 q_2 + p_1 q_2 + p_1 p_2 + p_2 q_1\bigr)}\;\pi_3.
    \end{align}
We conclude that $(\pi_1, \pi_2, \pi_3)$ given by \eqref{First_derivation_TVRS_second_example} and \eqref{Second_derivation_TVRS_second_example}, where $\pi_3$ can be any positive real number, is an actuarially fair Tavin RS scheme with the two participants described above.
\end{example}

\begin{example}
Consider the RS rule $\boldsymbol{P}$ based on
the RS schemes $\left(  \boldsymbol{\pi},\boldsymbol{P}\left[  \boldsymbol{\pi
}\right]  \right)  $ of Example 2, with constant number of protection units for each participant:
\[
f_{i}=1,\qquad i=1,2,\ldots,n+1.
\]
This means that the RS rule $\boldsymbol{P}$ is such that for any RS scheme
$\left(  \boldsymbol{\pi},\boldsymbol{P}\left[  \boldsymbol{\pi}\right]
\right)  $, the relative compensation vector $\boldsymbol{P}\left[  \boldsymbol{\pi
}\right]  $ is given by
\begin{equation}
P_{i}\left[  \boldsymbol{\pi}\right]  =\frac{I_{i}}{\sum_{j=1}^{n+1}I_{j}%
},\qquad i=1,2,\ldots,n+1.
\end{equation}
Obviously, $\boldsymbol{P}$ satisfies the indifference property (\ref{C61a}).
The RS scheme $\left(  \boldsymbol{\pi},\boldsymbol{P}\left[  \boldsymbol{\pi
}\right]  \right)  $ is actuarially fair for all participants if and only if
the conditions (\ref{C14}) are satisfied.\ Inspired by the approach proposed
in Denuit and Robert (2025), one can verify that the expected proportions
follow from%
\begin{equation}
E\left[  P_{i}\left[  \boldsymbol{\pi}\right]  \right]  =\sum_{k=1}^{n}%
\frac{1}{k}\Pr\left[  I_{i}=1,\text{ }\sum_{j=1}^{n}I_{j}=k\right]  ,\qquad
i=1,2,\ldots,n. \label{C53}%
\end{equation}
In case the indicator variables $I_{i}$ of the $n$ participants are mutually independent, the
expressions (\ref{C53}) can be transformed into%
\begin{equation}
E\left[  P_{i}\left[  \boldsymbol{\pi}\right]  \right]  =\Pr\left[
I_{i}=1\right]  \sum_{k=1}^{n}\frac{1}{k}\Pr\left[  \sum_{j=1}^{n}I_{j}%
-I_{i}=k-1\right]  \qquad i=1,2,\ldots,n. \label{C56}%
\end{equation}
In this case, the expectations of the relative compensations $P_{i}\left[
\boldsymbol{\pi}\right]  $ follow from probabilities of events related to sums
of independent Bernoulli random variables. The actuarially fair initial investments of
this RS scheme follow then from Proposition 2 and from (\ref{C56}).\ Notice
that for this RS rule, the actuarially fair initial investments are only
defined up to a positive constant factor.\ \newline The calculations above can
in a straightforward way be generalized to the case that all $f_{i}$ are
positive (not necessary equal) integers, rather than all equally 1.\ We refer
to Denuit and Robert (2025) for more details on this case.\ \hfill$\lhd$
\end{example}

\begin{example} \label{example:gamma_dirichlet_act_fair}
    Consider the RS rule $\boldsymbol{P}$ based on the RS scheme $(\pi, P[\pi])$ of Example \ref{example:continuous} with $n$ participants and an active administrator. Then we have that for $i \in\{1,\ldots,n\}$,
        $$ E[P_i]=p_i  \left(\sum_{j=1}^{n+1} \pi_j\right) \cdot \left( \sum_{\substack{S \subseteq \{1,\ldots,n\} \\ i \in S}} \frac{\alpha_i}{\sum_{j \in S} \alpha_j} \prod_{j \in S \setminus \{i\}} p_j \prod_{j \in \{1,\ldots,n\} \setminus S} q_j \right) $$
    since $ E[P_i \mid I_i = 1, S = s] = \frac{\alpha_i}{\sum_{j\in S} \alpha_j} $ for the claimant set $S$, where the notation $S \subseteq \{1,\ldots,n\}$ refers to any subset $S$ of $\{1,\ldots,n\}$, and the summation notation $\sum_{\substack{S \subseteq \{1,\ldots,n\} \\ i \in S}}$ ranges over all subsets $S$ of $\{1,\ldots,n\}$ containing $i$. From (\ref{C36}) it then follows that the actuarial fairness conditions for all participants can be written as follows: 
    \begin{align}
        \pi_i &= p_i  \left(\sum_{j=1}^{n+1} \pi_j\right) \cdot \left( \sum_{\substack{S \subseteq \{1,\ldots,n\} \\ i \in S}} \frac{\alpha_i}{\sum_{j \in S} \alpha_j} \prod_{j \in S \setminus \{i\}} p_j \prod_{j \in \{1,\ldots,n\} \setminus S} q_j \right), \qquad i = 1, \ldots, n, \label{eq:fair-i}\\
        \pi_{n+1} &= \left( \sum_{j=1}^{n+1} \pi_j \right) \cdot \left( \prod_{j=1}^n q_j \right). \label{eq:fair-admin}
    \end{align}
    This means that in Example \ref{example:continuous}, there is a closed form solution for the investments of all agents in the RS scheme.
\end{example}

\section{Compensation-based risk-sharing with a passive administrator}

\subsection{Introducing the passive administrator}

Consider an RS scheme $\left(  \boldsymbol{\pi},\boldsymbol{P}\right)  $ with
$n$ participants and an active administrator as defined above.\ The time $0$
investments of the $n+1$ agents are summarized in the investment vector
$\boldsymbol{\pi}=\left(  \pi_{1},\pi_{2},\ldots,\pi_{n+1}\right)  $, while
the relative compensation vector is given by $\boldsymbol{P}=\left(
P_{1},P_{2},\ldots,P_{n+1}\right)  $. The compensations are summarized in
the compensation vector $\boldsymbol{W}=\left(  W_{1},W_{2},\ldots
,W_{n+1}\right)  $.\ The latter vector is given by
\begin{equation}
W_{i}=\left(  \sum_{j=1}^{n+1}\pi_{j}\right)  \times P_{i}\qquad
i=1,2,\ldots,n+1. \label{C100}%
\end{equation}
This set up of a compensation-based RS scheme with an active administrator is
a generalization of the set up in Dhaene and Milevsky (2024), who consider
the RS schemes related to tontine funds, as described in Example 2. 

Denuit and Robert (2025) also consider the setting of Dhaene and Milevsky
(2024), but with a passive administrator who does not join the group of
participants in investing. Instead, the administrator is present only to oversee participants' compliance and facilitate smooth operation of the fund. Hence, these authors assume that $\pi_{n+1}=0$,
implying that (without any adaptation of the compensations), such an RS scheme
can never be actuarially fair, see Remark 1. In order to
include the possibility that the RS scheme is actuarially fair for all
participants, Denuit and Robert (2025) assume that in case no participant
receives a positive compensation, that is, in case $\sum_{j=1}^{n}P_{j}=0$, or
equivalently, $P_{n+1}=1$, each participant $i$ receives back his original
investment. Most papers on tontine funds in the last few years have added this
element to repair expectations and hence, actuarial
fairness. Hereafter, we will generalize this approach considered in Denuit
and Robert (2025), in the same way as we generalized the approach of Dhaene
and Milevsky (2024) in the first part of this paper. That is, we will consider
risk-sharing with a passive administrator for a general class of
compensations, instead of restricting to the tontine fund case of Example
2.\ Let us start by defining a compensation-based RS scheme with a passive administrator.\ \ 

\begin{definition}
A compensation-based risk-sharing scheme for a given group of $n$ participants
and a passive administrator is a pair $\left(  \boldsymbol{\pi},\boldsymbol{P}%
\right)  $, where $\boldsymbol{\pi}=\left(  \pi_{1},\pi_{2},\ldots,\pi
_{n+1}\right)  $ is the initial investment vector with $\pi_{n+1}=0$. Further, $\boldsymbol{P}%
=\left(  P_{1},P_{2},\ldots,P_{n},P_{n+1}\right)  $ is the relative compensation
vector, that is $\boldsymbol{P}\in\mathcal{R}_{n+1}$ defined in (\ref{C6d}), and where the
compensations attributed to the $n$ participants follow from the compensation
vector $\boldsymbol{W}=\left(  W_{1},W_{2},\ldots,W_{n+1}\right)  $,
which is defined by
\begin{equation}
W_{i}= \left(  \sum_{j=1}^{n}\pi_{j}\right)  \times P_{i}+\pi_{i}\times P_{n+1} \qquad :i=1,2,\ldots,n,
\label{C1a}%
\end{equation}
whereas $W_{n+1}=0$.
\label{def6}
\end{definition}

From this definition, we see that an RS scheme $\left(  \boldsymbol{\pi
},\boldsymbol{P}\right)  $ with a passive administrator is one where each
participant $i$ invests an an amount $\pi_{i}$ and receives a relative compensation
$P_{i}$ of the available fund at time $1$, whereas the administrator makes an
investment $\pi_{n+1}=0$ and receives a zero-compensation $W_{n+1}=0$ at time
$1$.\ Moreover, in case the relative compensations $P_{i}$ of all participants
are $0$, or equivalently, $P_{n+1}=1$, every participant receives his original
investment $\pi_{i}$ back.\ The administrator is `passive' in the sense that
he is not involved in investing and receiving any amount of compensation based on the participant's initial investments.\ 

In order to be able to clearly distinguish between the 'active' and the 'passive' administrator case, see (\ref{C100}) and (\ref{C1a}), hereafter we will sometimes write $W_i^{\text{active}}$ or $W_i^{\text{passive}}$ instead of $W_i$ for the contribution of agent $i$.

Taking into account that $\sum_{j=1}^{n+1}P_{j}=1$, it is a straightforward
exercise to prove the following full allocation property for a
compensation-based RS scheme with a passive administrator:
\begin{equation}
\sum_{i=1}^{n}W_{i}=\sum_{j=1}^{n}\pi_{j}.
\end{equation}
This means that the total amount of compensations paid at time $1$ is exactly equal
to the fund value at that time. Hereafter, when considering an `RS scheme
$\left(  \boldsymbol{\pi},\boldsymbol{P}\right)  $ with a passive
administrator', we mean a `compensation-based RS scheme with a passive administrator'.\ 

\begin{example}
Denuit and Robert (2025) and Dhaene and Milevsky (2024) discuss the
above-mentioned approach with a passive administrator for the setting of
Example 2, that is, in the framework of a tontine fund with $\pi_{n+1}=0$. In particular, they
consider the RS scheme $\left(  \boldsymbol{\pi},\boldsymbol{P}\right)  $ with
a passive administrator, where the compensation vector $\boldsymbol{W}$ is
defined by (\ref{C1a}), with relative compensation vector $\boldsymbol{P}$
given by%
\begin{equation}
P_{i}=\frac{f_{i}\times I_{i}}{\sum_{j=1}^{n+1}f_{j}\times I_{j}},\qquad
i=1,2,\ldots,n+1,
\end{equation}
where the protection units $f_{i}$ and the indicator variables $I_{i}$ are
defined in Example 2. Dhaene and Milevsky (2024) write the following about
this approach: `While the above-mentioned approach (i.e., without an (active)
administrator, and returning back the investments if no participant
receives a positive compensation) resolves the mathematical problem, we
believe that this isn't why people buy tontines.\ Indeed, it violates the
spirit of the (historical) tontine in which all rights and ownership benefits
are lost at death.\ Furthermore, some members may not have any beneficiaries,
leading to yet another unintended redistribution of wealth.\ In extreme cases,
when there is only one person surviving, this may create a moral hazard.\ In
other words, and for many reasons, while adding a death benefit refund or
payout `solves' the math, it `ruins' the elegance of the tontine ideal.'\ As
far as we are aware, Dhaene and Milevsky (2024) are the first who introduce an
active tontine administrator as both a technical and real-world solution to
some of the above-mentioned issues, instead of artificially adding legacy or
bequest payouts to participants.\ 
\end{example}

\subsection{Actuarially fair risk-sharing schemes with a passive administrator}
Similar to actuarial fairness for an RS scheme with an active administrator,
the RS scheme $\left(  \boldsymbol{\pi},\text{ }\boldsymbol{P}\right)  $ with
a passive administrator is said to be actuarially fair for the participants if
the time $1$ value of each participant's initial investment is equal to the
expected value of the compensation he will receive at time $1$.\ 

\begin{definition}
The RS scheme $\left(  \boldsymbol{\pi},\boldsymbol{P}\right)  $ with a
passive administrator is actuarially fair for all participants if the following
conditions hold:
\begin{equation}
\pi_{i}=E\left[  W_{i}\right]  ,\qquad i=1,2,\ldots,n. \label{C45c}%
\end{equation}

\end{definition}

Taking into account that $\pi_{n+1} = W_{n+1}=0$, we find that the RS scheme is always actuarially fair for the passive administrator.

In the next proposition, another characterization is given for the actuarial
fairness of an RS scheme with a passive administrator.\ 

\begin{proposition}
The RS scheme $\left(  \boldsymbol{\pi},\boldsymbol{P}\right)  $ with a
passive administrator is actuarially fair for its $n$ participants if and only if any of the following conditions is satisfied:
\vspace{1mm}
\newline\textbf{Condition 1}: The RS scheme $\left(  \boldsymbol{\pi
},\boldsymbol{P}\right)  $ satisfies%
\begin{equation} 
\pi_{i}=\left(  \sum_{j=1}^{n}\pi_{j}\right)  \times E\left[
P_{i}\right] + \pi_i \times E[P_{n+1}] ,\qquad i=1,2,\ldots,n.
\label{C45e}%
\end{equation}
\newline\textbf{Condition 2}: The RS scheme $\left(  \boldsymbol{\pi},\boldsymbol{P}\right)  $ satisfies
\begin{equation}
\pi_{i}=\left(  \sum_{j=1}^{n}\pi_{j}\right)  \times\frac{E\left[  P_{i}\right]}
{\Pr\left[  P_{n+1}=0\right]  }, \qquad i=1,2,\ldots,n.
\label{C45d}%
\end{equation}
\end{proposition}

\begin{proof}
Actuarial fairness for all participants is, by definition, equivalent to Condition
1. Taking into account (\ref{C31}), the actuarial fairness
conditions (\ref{C45e}) can be transformed into the expressions
(\ref{C45d}).
\end{proof}

In the following corollary, we show that under appropriate conditions, actuarially fair investments $\pi_i$ of the participants in a RS scheme $(\boldsymbol{\pi},\boldsymbol{P})$ with an active administrator are also actuarially fair in the corresponding RS scheme $(\boldsymbol{\pi},\boldsymbol{P})$ with a passive administrator.

\begin{corollary}
     Consider the RS scheme $\left( (\pi_1, \ldots,\pi_n,\pi_{n+1}), \boldsymbol{P} \right)$ with an active administrator (as defined in $(\ref{C100})$) and the RS scheme $\left( (\pi_1, \ldots,\pi_n,0), \boldsymbol{P} \right)$ with a passive administrator (as defined in $(\ref{C1a})$).  Then we have that the following statements are equivalent:
    \begin{enumerate}[label=(\alph*)]
        \item $ \pi_i =  E[{W_i^{\text{active}}}], \hspace{2mm} i=1,2,\ldots,n+1 $
        \vspace{2mm}
        \item $ \pi_i = E[{W_i^{\text{passive}}}], \hspace{2mm} i=1,2,\ldots, n$ \hspace{2mm} and \hspace{2mm} $ \pi_{n+1} = E[W_{n+1}^{\text{active}}]. $
    \end{enumerate}
\end{corollary}

\begin{proof}
    Taking into account Proposition 1 and Condition 2 in Proposition 2, we can rewrite the equations $(a)$ as follows:
    $$ \pi_i = \left( \sum_{j=1}^n \pi_j \right) \times \frac{E[P_i]}{\Pr[P_{n+1}=0]}, \qquad i=1,2,\ldots,n  $$
    and
    $$ \pi_{n+1} = \left( \sum_{j=1}^{n+1} \pi_j \right) \times  E[P_{n+1}]. $$
    From Condition 2 in Proposition 5, we can conclude the equations in $(a)$ are equivalent to the equations in $(b)$.

\end{proof}

Notice that the only if statement of corollary above can also be proven by rewriting (\ref{C100}) and (\ref{C1a}) as 
$$ W^{\text{active}}_i =  \left( \sum_{j=1}^{n} \pi_j \right) \times P_i + \pi_{n+1} \times P_i, \qquad i=1,2,\ldots,n+1 $$
and 
$$ W^{\text{passive}}_i =  \left( \sum_{j=1}^{n} \pi_j \right) \times P_i + \pi_{i} \times P_{n+1}, \qquad i=1,2,\ldots,n. $$

When the participants' contributions are actuarially fair in case of an active administrator, then from Condition 3 of Proposition 2, we have that

$$ \pi_{n+1} \times E[P_i] = \pi_i \times E[P_{n+1}], \qquad i=1,\ldots,n+1. $$
This implies that 
$$ E[W^{\text{active}}_i] = E[W^{\text{passive}}_i], \qquad i=1,\ldots, n, $$
which proves the stated implication.

Suppose that the participants in an RS scheme decide to first fix their initial investments and only afterwards make the choice between an active or a passive administrator. Corollary 1 states that from an actuarial fairness point of view, both schemes are equivalent, provided the active administrator makes an actuarially fair initial investment.

\subsection{Risk-sharing rules with a passive administrator}
Similar to section 4.2, where we considered RS rules with an active administrator, we introduce RS rules with a passive administrator. An RS rule $\boldsymbol{P}$ with a passive administrator transforms any investment vector $\boldsymbol{\pi}$ into a relative compensation vector $\boldsymbol{P}[\boldsymbol{\pi}]$. This means that any investment vector leads to the compensation vector $\boldsymbol{W}[\boldsymbol{\pi}]$, with
\begin{equation}
W_i[\boldsymbol{\pi}]= \left( \sum_{j=1}^{n} \pi_{j} \right) \times P_i \left[ \boldsymbol{\pi}\right] + \pi_i \times P_{n+1} \left[ \boldsymbol{\pi} \right], \qquad i=1,2,\ldots, n.
\end{equation}
As before, we consider the indifference property (\ref{C61a}) for RS rules with a passive administrator.
It is then a straightforward exercise to prove that Propositions 3 and 4 can easily be adapted to hold for RS rules with a passive administrator.
This means that for RS rules with a passive administrator, which satisfy the indifference property (\ref{C61a}), actuarially fair investment vectors are only defined up to a positive constant factor.

\section{The two participants tontine fund}

In this section, we investigate an RS scheme related to a tontine fund with
only two participants.\ We consider both cases of an active and a passive
administrator.\ We investigate actuarial fairness conditions for the heterogeneous case
where survival probabilities of the two participants may be different.\ 

\subsection{The two participants tontine fund with an active administrator}

Consider two participants who set up a
tontine fund. The initial investment made by participant $i$, $i=1,2,$ is
denoted by $\pi_{i}$, while his survival probability is given by
$p_{i}=1-q_{i}$. In addition to the two participants, also a third agent,
called the active administrator, is involved. His contribution to the
investment pool is denoted by $\pi_{3}$. As before, for simplicity, we assume a
zero return over the observation period.

Suppose that the participants agree on an RS rule with an active administrator,  with relative compensation vector as described in Example 2, that is%
\begin{equation}
P_{i}=\frac{f_{i}\times I_{i}}{\sum_{j=1}^{3}f_{j}\times I_{i}},\qquad
i=1,2,3,\label{C81a}%
\end{equation}
and compensations determined from
\begin{equation}
W_{i}=\left(  \pi_{1}+\pi_{2}+\pi_{3}\right)  \times P_{i},\qquad
i=1,2,3.\label{C81}%
\end{equation}
The protection units $f_{i}$ and the indicator variables $I_{i}$ are as
defined in Example 2. In particular, we have that $I_{i}$ is equal to $1$ in
case participant $i=1,2$ survives, while it is equal to $0$ otherwise.\ Furthermore,
$I_{3}$ is equal to $\left(  1-I_{1}\right)  \times\left(  1-I_{2}\right)  $.

The relative compensation vector $\boldsymbol{P}=\left(  P_{1},P_{2}%
,P_{3}\right)  $ of the RS scheme $\left(  \boldsymbol{\pi},\boldsymbol{P}%
\right)  $ defined by (\ref{C81a}) is assumed to be independent of the initial
investments in the sense that the $f_i$'s are given real numbers, independent of $\pi_i$. This vector can be expressed as follows:
\begin{equation}
\boldsymbol{P}=\left\{
\begin{array}
[c]{cc}%
\left(  1,0,0\right)   & :\text{if }I_{1}=1\text{ and }I_{2}=0\\
\left(  0,1,0\right)   & :\text{if }I_{1}=0\text{ and }I_{2}=1\\
\left(  0,0,1\right)   & :\text{if }I_{1}=0\text{ and }I_{2}=0\\
\left(  \beta,1-\beta,0\right)   & :\text{if }I_{1}=1\text{ and }I_{2}=1
\end{array}
\right.  \label{A4}%
\end{equation}
with $\beta$ given by%
\begin{equation}
\beta=\frac{f_{1}}{f_{1}+f_{2}}.\label{C80}%
\end{equation}
If participant 1 survives and participant 2 dies, the total amount of $\left(
\pi_{1}+\pi_{2}+\pi_{3}\right)  $ is awarded to person 1 at time $1$.
Similarly, if participant 1 dies while participant 2 survives, the total
amount is awarded to person 2 at time $1$. If no participant survives, the
total proceeds $\left(  \pi_{1}+\pi_{2}+\pi_{3}\right)  $ belong to the
administrator. If both participants survive, $\left(  \pi_{1}+\pi_{2}+\pi
_{3}\right)  $ is shared by the two participants:\ The first participant takes
a relative share $\beta$ of the available funds, while the second participant
receives a relative share $\left(  1-\beta\right)  $\ of these funds.\ 

When setting up the tontine RS scheme, the two participants and the
administrator must agree on the value of the relative share $\beta$ of the
total fund that participant $1$ will receive if both participants survive, as
well as on the investments $\pi_{1}$, $\pi_{2}$ and $\pi_{3}$. Given the
relative share $\beta$, the choice of the initial investments by the
participants may reflect considerations about their survival
probabilities. Consider, for example, uniform risk-sharing, that is $\beta=\frac{1}{2}$ or,
equivalently, $f_{1}=f_{2}$, then equal investments may be considered as
`unfair' as this choice does not take into consideration that health (survival
probabilities) may be different for both participants. However, notice that
once the relative share $\beta$ and the investments are chosen, knowledge of
the survival probabilities is no longer required to be able to further manage
the tontine fund.\ 

We assume that the remaining lifetimes of the two participants are mutually
independent.\ The probabilities of the different relevant events related to
the tontine fund payouts can then be expressed as
\begin{equation}
\Pr\left[  I_{1}=i\text{ and }I_{2}=j\right]  =\left\{
\begin{array}
[c]{cc}%
p_{1}\times q_{2} & :\text{if }i=1\text{ and }j=0\\
q_{1}\times p_{2} & :\text{if }i=0\text{ and }j=1\\
q_{1}\times q_{2} & :\text{if }i=0\text{ and }j=0\\
p_{1}\times p_{2} & :\text{if }i=1\text{ and }j=1.
\end{array}
\right.  \label{A4a}%
\end{equation}

In this subsection, we want to find out what is a
reasonable choice for the investments $\pi_{1}$, $\pi_{2}$ and $\pi_{3}$, once
the relative share $\beta$ is chosen. Remark that in case `reasonability' is
translated into `actuarial fairness', the participants have to agree on a common
choice for the survival probabilities $p_{1}$ and $p_{2}$ to determine the
compensations (also called the tontine fund payouts) $W_{i}$.

Taking into account (\ref{A4}) and (\ref{A4a}), we find that the expected
relative compensations $E\left[  P_{i}\right]  $ of the participants are given
by%
\begin{equation}
E\left[  P_{1}\right]  =p_{1}\times\left(  q_{2}+\beta\times p_{2}\right)
\end{equation}
and%
\begin{equation}
E\left[  P_{2}\right]  =p_{2}\times\left(  q_{1}+\left(  1-\beta\right)
\times p_{1}\right)  ,
\end{equation}
while the administrator's expected relative compensation $E\left[
P_{3}\right]  $ is given by
\[
E\left[  P_{3}\right]  =q_{1}\times q_{2}.
\]

Let us now determine the actuarially fair investments $\pi_{1},\pi_{2}$ and $\pi
_{3}$ for the three participants.\ From (\ref{C15c}), we find that the actuarial
fairness conditions, which state that each agent's investment is equal to his
expected compensation, are given by
\begin{equation}
\left\{
\begin{array}
[c]{lll}%
\pi_{1} & = & \left(  \pi_{1}+\pi_{2}+\pi_{3}\right)  \times p_{1}%
\times\left(  q_{2}+\beta\times p_{2}\right) \\
\pi_{2} & = & \left(  \pi_{1}+\pi_{2}+\pi_{3}\right)  \times p_{2}%
\times\left(  q_{1}+\left(  1-\beta\right)  \times p_{1}\right) \\
\pi_{3} & = & \left(  \pi_{1}+\pi_{2}+\pi_{3}\right)  \times q_{1}\times q_{2}.%
\end{array}
\right.  \label{C82}%
\end{equation}
\noindent\noindent\noindent\ Any solution $\left(  \pi_{1},\pi_{2},\pi
_{3}\right)  $ of (\ref{C82}) is a set of actuarially fair initial
investments.\ Determining such a set of investments requires a common choice
for the survival probabilities of the two participants.\ Important to note is that
even if a `wrong' choice (for instance a too conservative choice) is made,
there will not be an issue of insolvency, due to the full allocation condition
(\ref{C1}), which guarantees that the sum of all tontine fund payouts
(compensations) is exactly equal to the available funds.\ 

Taking into account (\ref{C15b}) and the fact that $\Pr\left[  P_{3}=0\right]
=1-q_{1}\times q_{2}$, we can rewrite the actuarial fairness conditions
(\ref{C82}) as follows:
\begin{equation}
\left\{
\begin{array}
[c]{lll}%
\pi_{1} & = & \left(  \pi_{1}+\pi_{2}\right)  \times p_{1}\times\frac
{q_{2}+\beta\times p_{2}}{1-q_{1}\times q_{2}}\\
\pi_{2} & = & \left(  \pi_{1}+\pi_{2}\right)  \times p_{2}\times\frac
{q_{1}+\left(  1-\beta\right)  \times p_{1}}{1-q_{1}\times q_{2}}\\
\pi_{3} & = & \left(  \pi_{1}+\pi_{2}\right)  \times\frac{q_{1}\times q_{2}%
}{1-q_{1}\times q_{2}}%
\end{array}
\right.  \label{C83}%
\end{equation}

From (\ref{C15a}), we find that the actuarially fairness conditions
can also be expressed as follows:%
\begin{equation}
\left\{
\begin{array}
[c]{l}%
\pi_{1}=\pi_{3}\times p_{1}\times\frac{q_{2}+\beta\times p_{2}}{q_{1}\times
q_{2}}\\
\pi_{2}=\pi_{3}\times p_{2}\times\frac{q_{1}+\left(  1-\beta\right)  \times
p_{1}}{q_{1}\times q_{2}}%
\end{array}
\right.  \label{C84}%
\end{equation}

The formulas (\ref{C82}), (\ref{C83}) and (\ref{C84}) provide 3 different ways
to determine actuarially fair initial investments, once there is an agreement on
the relative share $\beta$ and the survival probabilities $p_{i}$.\ One can
first choose the aggregate investments $\left(  \pi_{1}+\pi_{2}+\pi
_{3}\right)  $ of all participants, or the total investments $\left(  \pi_{1}%
+\pi_{2}\right)  $ of all participants, or the investment $\pi_{3}$ of the
administrator, and then use either (\ref{C82}), (\ref{C83}) or (\ref{C84}) to
derive the actuarially fair investments for the 3 participants.$\ $

\begin{example}
\textbf{(A game of chance with a coin and a die)}\ Consider the combined game
of chance with two participants, described in the online appendix of Dhaene and
Milvesky (2024). To enter the game, participant 1 pays an amount $\pi_{1}$,
while participant 2 pays $\pi_{2}$. Participant 1 tosses a two-sided coin,
while participant 2 rolls a six-sided die. In this game, participant 1 is
successful if he tosses heads, while participant 2 is successful if he rolls a
1. In addition to the two participants, an active administrator is involved.
He contributes to the prize pool by paying an amount $\pi_{3}$. \newline The
payouts for this game of chance are defined as follows:\ If the coin lands on
heads and the die does not land on 1, the total amount of $\left(  \pi_{1}%
+\pi_{2}+\pi_{3}\right)  $ is awarded to participant 1. Similarly,
if the coin does not land on heads but the die lands on 1, the total amount of
$\left(  \pi_{1}+\pi_{2}+\pi_{3}\right)  $ is awarded to participant 2. If both participants are successful (i.e., heads and 1 appear after the
respective throws), the total proceeds of $\left(  \pi_{1}+\pi_{2}+\pi
_{3}\right)  $ are shared by person 1 and person 2.\ In this case, participant
1 receives a relative share $\beta$, while participant 2 receives a relative
share $\left(  1-\beta\right)  $ of the available fund. Finally, if both
participants are not successful (i.e., neither heads nor 1 appear after their
respective throws), the total proceeds of $\left(  \pi_{1}+\pi_{2}+\pi
_{3}\right)  $ go to the administrator.\ Let us assume that the outcomes of the
coin and the die are independent.\ The probability that the administrator
will receive the entire prize pool is then given by $\frac{5}{12}$, so it
seems reasonable to require the administrator to contribute to the prize pool
for his chance of winning.


This game of chance can be described
within the framework of the two participants tontine fund considered above in
this subsection:%
\[
W_{i}=\left(  \pi_{1}+\pi_{2}+\pi_{3}\right)  \times P_{i}\qquad i=1,2,3,
\]
with the relative compensations $P_{i}$ defined by (\ref{A4}).\newline From
(\ref{C82}), we find that the game of chance is actuarially fair for the 3
participants if and only if the amounts $\pi_{1},\pi_{2}$ and $\pi_{3}$ satisfy the
following set of equations:
\begin{equation}
\left\{
\begin{array}
[c]{lll}%
\pi_{1} & = & \left(  \pi_{1}+\pi_{2}+\pi_{3}\right)  \times\frac{5+\beta}%
{12}\\
\pi_{2} & = & \left(  \pi_{1}+\pi_{2}+\pi_{3}\right)  \times\frac{2-\beta}%
{12}\\
\pi_{3} & = & \left(  \pi_{1}+\pi_{2}+\pi_{3}\right)  \times\frac{5}{12}%
\end{array}
\right.  \label{C95a}%
\end{equation}
In this case, the actuarially fair initial investments $\pi_{1},\pi_{2}$ and
$\pi_{3}\ $are expressed as proportions of the total payment $\left(  \pi
_{1}+\pi_{2}+\pi_{3}\right)  $ of the three participants.\ One can choose the
magnitude of $\left(  \pi_{1}+\pi_{2}+\pi_{3}\right)  $ first, and then
determine the corresponding actuarially fair initial payments of the
participants and the administrator by (\ref{C95a}).\ \newline From
(\ref{C83}), we find that the actuarially fairness conditions can also be
expressed in the following way:
\begin{equation}
\left\{
\begin{array}
[c]{lll}%
\pi_{1} & = & \left(  \pi_{1}+\pi_{2}\right)  \times\frac{5+\beta}{7}\\
\pi_{2} & = & \left(  \pi_{1}+\pi_{2}\right)  \times\frac{2-\beta}{7}\\
\pi_{3} & = & \left(  \pi_{1}+\pi_{2}\right)  \times\frac{5}{7}%
\end{array}
\right.  \label{C95}%
\end{equation}
In this case, the actuarially fair investments $\pi_{1},\pi_{2}$ and $\pi_{3}\ $are
expressed as proportions of the total payment $\left(  \pi_{1}+\pi_{2}\right)
$ of the two participants.\ One can choose the magnitude of $\left(  \pi
_{1}+\pi_{2}\right)  $ first , and then determine the corresponding initial
payments of the participants and the administrator by (\ref{C95}%
).\ Alternatively, one can first choose $\pi_{1}$ and $\pi_{2}$, and then
determine $\pi_{3}$ and set $\beta$ equal to $\frac{2\pi_{1}-5\pi_{2}}{\pi
_{1}+\pi_{2}}$, such that the game is actuarially fair for all participants.
However, notice that in this case, not every choice of $\pi_{1}$ and $\pi_{2}$
will lead to a relative share $\beta$ between $0$ and $1$.\newline Finally,
from (\ref{C84}) one finds that the actuarially fair initial investments also follow from%
\begin{equation}
\left\{
\begin{array}
[c]{c}%
\pi_{1}=\pi_{3}\times\frac{5+\beta}{5}\\
\pi_{2}=\pi_{3}\times\frac{2-\beta}{5}%
\end{array}
\right.  \label{C96}%
\end{equation}
Here, $\pi_{1}$ and $\pi_{2}$ are expressed as proportions of the investment
$\pi_{3}$ of the administrator.\ One can choose the magnitude of $\pi_{3}$
first, and then determine the investments of the two participants by
(\ref{C96}).\ \hfill$\lhd$
\label{coinanddie}
\end{example}

\subsection{The two participants tontine fund with a passive administrator}

Let us now replace the tontine RS scheme with an active administrator
defined in the previous subsection by the tontine RS scheme $\left(
\boldsymbol{\pi},\boldsymbol{P}\right)  $ with a passive administrator, as
explained before.\ The investment vector and the compensation  vector are now
given by $\boldsymbol{\pi}=\left(  \pi_{1},\pi_{2},0\right)  $ and
$\boldsymbol{W}=\left(  W_{1},W_{2},0\right)  $. From (\ref{C1a}), it follows
that the compensations $W_{i}$ of the participants are given by

\begin{equation}
W_{i}= 
\left(  \pi_{1}+\pi_{2} \right)  \times P_{i}+\pi_{i}\times P_{3},\qquad i=1,2,
\end{equation}
with the relative compensation vector $\boldsymbol{P}=\left( P_{1},P_{2},P_{3}\right)  $ given by (\ref{A4}).

From (\ref{C45d}), we find that the actuarial fairness conditions for the two
participants are given by
\begin{equation}
\left\{
\begin{array}
[c]{lll}%
\pi_{1} & = & \left(  \pi_{1}+\pi_{2}\right)  \times\frac{p_{1}\times\left(
q_{2}+\beta\times p_{2}\right)  }{1-q_{1}\times q_{2}}\\
\pi_{2} & = & \left(  \pi_{1}+\pi_{2}\right)  \times\frac{p_{2}\times\left(
q_{1}+\left(  1-\beta\right)  \times p_{1}\right)  }{1-q_{1}\times q_{2}}
\end{array}
\right.  \label{C97}%
\end{equation}
Expressions (\ref{C97}) provide a way to determine actuarially fair initial
investments for the two participants tontine fund with a passive
administrator.\ Indeed, once there is an agreement on the relative share
$\beta$ and the survival probabilities $p_{i}$, one can first choose the
aggregate initial payments $\left(  \pi_{1}+\pi_{2}\right)  $ of the two
participants, and then use (\ref{C97}) to derive actuarially fair investments
$\pi_{i}$. Another possibility consists of first choosing $\pi_{1}$, and then
determining $\pi_{2}$ from (\ref{C97}):%
\begin{equation}
\pi_{2}=\pi_{1}\times\frac{p_{2}\times\left(  q_{1}+\left(  1-\beta\right)
\times p_{1}\right)  }{p_{1}\times\left(  q_{2}+\beta\times p_{2}\right)  }.
\end{equation}

\begin{example}
\textbf{(A game of chance with a coin and a die)}\ Let us revisit the `coin
and die' game of chance with active administrator considered in Example
\ref{coinanddie}.\ We replace this game of chance with one with a passive administrator and
denote it by $\left(  \boldsymbol{\pi},\boldsymbol{P}\right)  $. From
(\ref{def6}), we have that the payouts (compensations) of the transformed RS
scheme are given by (\ref{C1a}) with the relative compensation vector
$\mathbf{P}$ defined by (\ref{A4}).\newline From (\ref{C97}), we find that
this game of chance is actuarially fair for all its participants if and only if the
initial investments $\pi_{1}$ and $\pi_{2}$ follow from the following set of
equations:%
\begin{equation}
\left\{
\begin{array}
[c]{lll}%
\pi_{1} & = & \left(  \pi_{1}+\pi_{2}\right)  \times\frac{5+\beta}{7}\\
\pi_{2} & = & \left(  \pi_{1}+\pi_{2}\right)  \times\frac{2-\beta}{7}.
\end{array}
\right.  \label{C95f}%
\end{equation}
This set of equations expresses the actuarially fair payments $\pi_{1}$ and
$\pi_{2}\ $as a proportion of the total investment $\left(  \pi_{1}+\pi
_{2}\right)  $ of the two participants.\ One can choose the magnitude of
$\left(  \pi_{1}+\pi_{2}\right)  $ first and then determine the investment
efforts of the participants by (\ref{C95f}).\ Alternatively, one can also
first choose $\pi_{1}$ and $\pi_{2}$, and then set $\beta$ equal to
$\frac{2\pi_{1}-5\pi_{2}}{\pi_{1}+\pi_{2}}$, so that\ the game is actuarial
fair.\ However, notice that in this case, not every choice of $\pi_{1}$ and
$\pi_{2}$ will lead to a relative share $\beta$ between $0$ and $1$%
.\hfill$\lhd$
\end{example}

\section{Comparison of utilities}
In this section, using the setup in Example \ref{example:continuous} and Example \ref{example:gamma_dirichlet_act_fair}, we compare the mean-variance utility for three cases: self-insurance, classical-insurance and compensation-based risk-sharing. \\
Consider the mean--variance utility under risk-aversion coefficient $\gamma_i > 0$, defined by $U_i(V_i) := E[V_i] - \tfrac{\gamma_i}{2}\cdot \text{Var}[V_i]$, where $V_i$ is the net position at time 1 (excluding deterministic initial wealth, which cancels in all comparisons). The mean-variance utility provides a simple basis for comparison, as it serves as a second-order approximation to any smooth, increasing, concave utility function.\\
\smallskip
\noindent\textbf{(SI) Self-insurance.} Participant $i$ bears $X_i = I_i Y_i$ in full: $V_i^{SI} = -X_i$, so
\begin{equation}
U_i(V_i^{SI}) = -E[X_i] - \tfrac{\gamma_i}{2} \cdot \text{Var}[X_i] = -p_i \alpha_i \theta - \tfrac{\gamma_i}{2} p_i \alpha_i \theta^2(1 + q_i \alpha_i), \label{eq:USI}
\end{equation}
since $ E[X_i] = p_i \alpha_i \theta $ and $ \text{Var}[X_i] = p_i \alpha_i \theta^2(1 + q_i \alpha_i) $. \\ \\
\smallskip
\noindent\textbf{(C) Classical insurance with premium loading $\eta \geq 0$.} The insurer charges the individual $\pi_i^C(\eta) := (1 + \eta)E[X_i]$ in exchange for full indemnification. The net position $V_i^C = -\pi_i^C(\eta)$ is deterministic, so
\begin{equation}
U_i(V_i^{C}) = -(1 + \eta) E[X_i] = -(1 + \eta) \cdot p_i \alpha_i \theta. \label{eq:UC}
\end{equation} \\ \\
\smallskip
\noindent\textbf{(RS) Compensation-based risk-sharing.} Under actuarial fairness $\pi_i = E[W_i]$, the deterministic part of $V_i^{RS} = -\pi_i + W_i - X_i$ collapses, giving
\begin{equation}
U_i(V_i^{RS}) = -E[X_i] - \tfrac{\gamma_i}{2} \cdot \text{Var}[ W_i - X_i] = -p_i \alpha_i \theta - \tfrac{\gamma_i}{2} \cdot \text{Var}[ W_i - X_i], \label{eq:URS}
\end{equation} \\
The three utilities lie on a common scale, $-E[X_i] + \text{correction}$, where the correction is $-\tfrac{\gamma_i}{2} \text{Var}[X_i]$ for self-insurance, $-\eta E[X_i]$ for loaded classical insurance, and $-\tfrac{\gamma_i}{2} \text{Var}[W_i - X_i]$ for fair compensation-based risk-sharing. Setting the correction for the classical insurance case equal to the compensation-based risk-sharing correction gives us the premium loading at which classical insurance just matches compensation-based risk-sharing. Similarly setting it equal to the self-insurance correction gives the premium loading at which classical insurance just matches self-insurance. We obtain
\begin{equation}\label{eq:thresholds}
    \eta_i^{*} := \frac{\gamma_i \text{Var}[W_i-X_i]}{2 E[X_i]}, \qquad
    \eta_i^{**} := \frac{\gamma_i\,\text{Var}[X_i]}{2E[X_i]}.
\end{equation}
$\eta_i^{*}$ is the premium loading the insurer must charge above pure premium for the participant to find loaded classical insurance exactly as attractive as compensation-based risk-sharing. $\eta_i^{**}$ is the analogous threshold for self-insurance. \\ \\
These two thresholds compare classical insurance to each of the other two schemes, but they leave one comparison: CBRS versus self-insurance. That comparison does not depend on $\eta$ at all, since neither scheme involves a loading, and it is governed by a single structural quantity: the quality of the hedge provided by compensation-based risk-sharing against the underlying risk $X_i$. This is measured by the slope $\beta_i := \text{Cov}[W_i, X_i]/\text{Var}[W_i]$. As we will see, the same quantity also controls whether $\eta_i^{*} \le \eta_i^{**}$, that is, whether the loading admits any region at all in which compensation-based risk-sharing strictly dominates both alternatives. \\
\begin{theorem}\label{thm:preferred}
Under actuarial fairness $\pi_i = \mathbb{E}[W_i]$, for $i = 1, 2,\ldots, n+1,$:
\begin{enumerate}[label=\textup{(\roman*)}]
\item $U_i(V_i^{RS}) \geq U_i(V_i^{C}) \iff \eta \geq \eta_i^{*}$, with equality iff $\eta = \eta_i^{*}$.
\item $U_i(V_i^C) \geq U_i(V_i^{SI}) \iff \eta \leq \eta_i^{**}$.
\item $U_i(V_i^{RS}) \geq U_i(V_i^{SI}) \iff \beta_i \geq 1/2$.
\item $\eta_i^{*} \le \eta_i^{**}$, if and only if $\beta_i \ge 1/2$. When this holds and $\eta \in (\eta_i^{*}, \eta_i^{**})$, compensation-based risk-sharing is strictly preferred to both alternatives:
\begin{equation}\label{eq:chain}
U^{\mathrm{RS}}_i \;>\; U^{\mathrm{C}}_i(\eta) \;>\; U^{\mathrm{SI}}_i.
\end{equation}
\end{enumerate}
\end{theorem}
\begin{proof}
Parts (i) and (ii) follow by direct subtraction: $U_i(V_i^{RS}) - U_i(V_i^{C}) = \eta E[X_i] - \tfrac{\gamma_i}{2}\text{Var}[W_i-X_i]$. This quantity is non-negative if and only if $\eta \ge \eta_i^{*}$. Similarly, $U_i(V_i^{C}) - U_i(V_i^{SI}) = \tfrac{\gamma_i}{2}\text{Var}[X_i] - \eta E[X_i]$. This quantity is non-negative if and only if $\eta \le \eta_i^{**}$.\\ \\
For (iii), subtract (\ref{eq:USI}) from (\ref{eq:URS}) and expand $\text{Var}[W_i-X_i]$:
\begin{equation}\label{eq:var-gap}
\text{Var}[X_i] - \text{Var}[W_i-X_i] = (2\beta_i - 1) \text{Var}[W_i],
\end{equation}
so $U_i(V_i^{RS}) \ge U_i(V_i^{SI})$ if and only if this difference is non-negative, that is, $\beta_i \ge 1/2$. \\ \\
For (iv), the same identity \eqref{eq:var-gap} gives
\begin{equation*}
\eta_i^{**} - \eta_i^{*} = \frac{\gamma_i\bigl(\text{Var}[X_i] - \text{Var}[W_i-X_i]\bigr)}{2 E[X_i]} = \frac{\gamma_i (2\beta_i - 1) \text{Var}[W_i]}{2E[X_i]},
\end{equation*}
which has the same sign as $2\beta_i - 1$. When $\beta_i \ge 1/2$ and $\eta \in (\eta_i^{*}, \eta_i^{**})$, parts (i) and (ii) give the strict chain~\eqref{eq:chain}.
\end{proof}\\ \\
Theorem \ref{thm:preferred} characterizes the participant's preference structure for a fixed pool of size $n$. The theorem states that when the compensation $W_i$ is a sufficiently good hedge for the individual's loss $X_i$, occurring when $\beta_i \geq 0.5$, compensation-based risk-sharing is preferred over self-insurance. Moreover, with $\beta_i \geq 0.5$, there exists a range of $(\eta_i^{*},\eta_i^{**})$ for $\eta$  under which compensation-based risk-sharing is preferred over classical insurance also. Conversely, when $\beta <0.5$, the hedge is too weak: the residual variance $\text{Var}[W_i-X_i]$ exceeds $ \text{Var}[X_i] $, so compensation-based risk-sharing leaves the participant worse off. In short, the desirability of compensation-based risk-sharing reduces to a single condition on the hedge quality $\beta_i$.

\section{The homogeneous tontine fund}

\subsection{The homogenous tontine fund with an active administrator}

In this subsection, we investigate a special case of the tontine RS scheme
$\left(  \boldsymbol{\pi},\boldsymbol{P}\right)  $ with an active
administrator, described in Example 2.\  We consider $n$
participants and an active administrator.\ We assume that all protection
units $f_{i}$ are equal to $1$.\ The $n+1$ agents each invest an initial
amount $\pi_{i}$. The components of the relative compensation vector are given
by
\begin{equation}
P_{i}=\frac{I_{i}}{\sum_{j=1}^{n+1}I_{j}},\qquad i=1,2,\ldots,n+1.\label{C98}%
\end{equation}
The compensations $W_{i}$ of the $n+1$ agents follow then from (\ref{C1aa}%
).\ We consider a homogeneous tontine RS scheme, which means that we assume
that the indicator variables $I_{i}$ of the $n$ participants are i.i.d, with
$\Pr\left[  I_{i}=1\right]  =p=1-q$.

Due to symmetry reasons, we must have that $E\left[  P_{i}\right]  $ is equal
for all $n$ participants.\ Furthermore, we have that
\begin{equation}
E\left[  P_{n+1}\right]  =\Pr\left[  I_{n+1}=1\right]  =q^{n}.\label{C99a}%
\end{equation}
Taking into account that $\sum_{i=1}^{n+1}E\left[  P_{i}\right]  =1$ leads to
\begin{equation}
E\left[  P_{i}\right] = E\left[\frac{I_{i}}{\sum_{j=1}^{n+1}I_{j}}\right] =\frac{1-q^{n}}{n},\qquad i=1,2,\ldots,n.\label{C99}%
\end{equation}

\begin{remark}
    In Example \ref{example:gamma_dirichlet_act_fair}, for the homogeneous case, that is $p_i = p$ and $\alpha_i=\alpha$ for all $i=1,2,\ldots,n$, we have the same result as in (\ref{C99a}) and (\ref{C99}).
\end{remark}

According to (\ref{C15b}), the homogeneous tontine RS scheme $\left(  \boldsymbol{\pi
},\boldsymbol{P}\right)  $ is actuarially fair for all agents if and only if
\begin{equation} \label{C95b}
    \pi_{i}=\frac{1}{n}   \sum_{j=1}^{n}\pi_{j} ,\qquad i=1,2,\ldots,n
\end{equation}
and
\begin{equation} \label{C95c}
    \pi_{n+1} = \frac{q^n}{1-q^n}  \sum_{j=1}^{n} \pi_j  .
\end{equation}
This implies that all initial investments $\pi_{i}$ of the
participants are equal, which we denote hereafter by $\pi$. 
This result was to be expected because of the inherent symmetry of the problem. Actuarially fair initial investments can then be determined by first choosing the administrator's initial investment $\pi_{n+1}$ and then determining the participants' initial investments by (\ref{C95b}) and (\ref{C95c}).

\subsection{The homogeneous tontine fund with a passive administrator}

Let us now replace the homogeneous tontine RS scheme with an active
administrator, considered in Section 8.1, by the corresponding homogeneous
tontine RS scheme $\left(  \boldsymbol{\pi},\boldsymbol{P}\right)  $ with a
passive administrator.\ The compensations $W_{i}$ are then given by
(\ref{C1a}), with the relative compensations $P_{i}$ as defined in
(\ref{C98}).\newline From (\ref{C45d}), (\ref{C99a}) and (\ref{C99}), we find
that $\left(  \boldsymbol{\pi},\boldsymbol{P}\right)  $ is actuarially fair for
its $n$ participants if and only if all initial investments $\pi_{i}$ of the
participants are all equal:%
\[
\pi_{i}=\frac{1}{n}\sum_{j=1}^{n}\pi_{j},\qquad i=1,2,\ldots,n.
\] 
Again, this result was to be expected because of the inherent symmetry of the problem.\ 

\subsection{Comparing the homogeneous tontine fund with an active administrator and a centralized insurance approach}

Consider a group of $n$ persons (participants). The one-year survival of
person $i$ is described by the Bernoulli random variable $I_{i}$, which equals $1$ if he
survives until time $1$, while it equals $0$ otherwise. The group
is assumed to be homogeneous in the sense that all Bernoulli random variables are i.i.d. with a common survival probability $p$. Suppose that each person is willing to invest an initial amount $\pi$ which will entitle him to a survival benefit at time $1$.

The participants could opt for \textit{centralized insurance}, by each buying
a pure endowment with a premium $\pi$. Suppose that the premiums are pure
premiums. In case person $i$ chooses an insurance with pure premium $\pi$,
then, assuming a zero-discount rate, his payout $W_{i}^{c}$ at time $1$ is
given by
\begin{equation}
W_{i}^{c}=\frac{\pi}{p}\times I_{i},\label{C86}%
\end{equation}
where we use the superscript `$c$' to indicate that this is the payout in a
`centralized' approach. The expected payoff in case of buying insurance is
given by
\[
E\left[  W_{i}^{c}\right]  =\pi.
\]
On average, the premium is equal to the expected insurance payment.\ This
means that the insurance approach is actuarially fair for each person.

Instead of buying a pure endowment insurance, the participants could decide to
take part in a homogeneous \textit{RS scheme} $\left(  \boldsymbol{\pi
},\boldsymbol{P}\right)  $ with an active administrator.\ Any participant
makes an initial investment $\pi$ in the tontine fund.\ The active
administrator's investment is, as usual, denoted by $\pi_{n+1}$.\ The initial investment vector is given by $\boldsymbol{\pi}=\left(
\pi,\pi,\ldots,\pi,\pi_{n+1}\right)  $.\ Suppose that the relative compensation vector $\boldsymbol{P}$ of the homogeneous RS scheme is
given by (\ref{C98}).\ The compensation $W_{i}(n)$ that participant $i$ will receive at time $1$ in a fund with $n$ participants in a decentralized approach with an active administrator is then given by
\begin{equation}
W_{i}(n)=\left(  n\times\pi+\pi_{n+1}\right)  \times\frac{I_{i}}{\sum
_{j=1}^{n+1}I_{j}},\qquad i=1,2,\ldots,n+1,\label{C86a}%
\end{equation}
where we added `$n$' in the notation $W_i(n)$ to indicate the number of participants in the RS scheme.

Let us now assume that the RS scheme $\left(  \boldsymbol{\pi},\boldsymbol{P}%
\right)  $ is actuarially fair for all participants and for the administrator.\ Then from (\ref{C95b}) and (\ref{C95c}) we find that the compensations $W_{i}(n)$ can be written as follows: \
\begin{equation}
W_{i}(n)=\frac{n \times \pi}{1-q^{n}}\times\frac{I_{i}}{\sum_{j=1}^{n+1}I_{j}},\qquad i=1,2,\ldots,n+1.\label{C87a}%
\end{equation}


From (\ref{C87a}), one can prove that for any participant we have 
$$ W_{i}(n)=\frac{ \pi}{1-q^{n}}\times\frac{I_{i}}{\frac{1}{n}\sum_{j=1}^{n}I_{j} + \frac{1}{n}\prod_{j=1}^n (1-I_j)} \xrightarrow{a.s.} \frac{\pi}{p} \times I_i. $$
Intuitively, this means that when the number of participants in the homogeneous tontine fund with an active administrator becomes infinitely large, then with probability $1$, the compensation of each participant $i$ becomes equal to the corresponding insurance payment $W_i^c$. 

Similarly, for the contribution of the administrator we have 
$$ W_{n+1}(n) = \frac{n \times \pi }{1-q^n} \times \prod_{j=1}^n (1-I_j)  \xrightarrow{a.s.} 0. $$

\noindent This means that when the number of participants in the homogeneous tontine fund with an active administrator becomes infinitely large, then with probability $1$, the contribution of the administrator becomes equal to $0$. We can conclude that the homogeneous tontine fund approach with an active administrator converges to the central insurance approach when the number of participants goes to infinity.

\subsection{Comparing the homogeneous tontine fund with a passive administrator and a centralized insurance approach}
Consider again the group of $n$ persons as described in Section 8.3. Suppose that the participants decide to take part in a homogeneous \textit{RS scheme} $\left(  \boldsymbol{\pi},\boldsymbol{P}\right)  $ with a passive administrator. Assuming again actuarially fair (and hence equal) contributions, any participant makes an initial investment, denoted by $\pi$, in the tontine fund. The initial investment vector is then given by $\boldsymbol{\pi}=\left(
\pi,\pi,\ldots,\pi,0\right)  $. The compensation that participant $i$ will receive at time $1$ in a fund with $n$ participants in a decentralized approach with a passive administrator then follows from
\begin{equation}
W_{i}(n)=n\times\pi  \times\frac{I_{i}}{\sum
_{j=1}^{n+1}I_{j}} + \pi \times I_{n+1} ,\quad i=1,2,\ldots,n.\label{C102}%
\end{equation}

From (\ref{C102}), for any participant $i$ one has that 

$$ W_i(n) = \pi  \times\frac{I_{i}}{\frac{1}{n}\sum
_{j=1}^{n}I_{j} + \frac{1}{n} \prod_{j=1}^n (1-I_j)} + \pi \times \prod_{j=1}^n\left( 1 - I_j \right)  \xrightarrow{a.s.} \frac{\pi}{p} \times I_i, $$

\noindent which means that when the number of participants in the homogeneous tontine fund with a passive administrator becomes infinitely large, then with probability $1$, the contribution of participant $i$ becomes equal to the insurance payment $W_i^c$. We can conclude that the tontine fund with a passive administrator also converges to the centralized insurance approach.

\section{Conclusion}

In this paper, we considered compensation-based RS schemes $\left(
\boldsymbol{\pi},\boldsymbol{P}\right)  $ with an active administrator.\ At
time $0$, each participant $i$ invests an initial amount $\pi_{i}$ in the
fund, while the administrator invests the amount $\pi_{n+1}$ in the same
fund.\ These initial investments are summarized in the investment vector
$\boldsymbol{\pi}=\left(  \pi_{1},\pi_{2},\ldots,\pi_{n+1}\right)  $. At time $1$,
each participant and also the administrator receives compensation from the
fund, summarized in the compensation vector $\boldsymbol{W}=\left(
W_{1},W_{2},\ldots,W_{n+1}\right)  $, which is given by
\[
W_{i}=\left(  \sum_{j=1}^{n+1}\pi_{j}\right)  \times P_{i}\qquad
i=1,2,\ldots,n+1,
\]
where $\boldsymbol{P}=\left(  P_{1},P_{2},\ldots,P_{n+1}\right)  $ is the
relative compensation vector of the RS scheme under consideration. This setup
is a generalization of the setup in Dhaene and Milevsky (2024) and Denuit and
Robert (2025), who consider RS schemes with relative compensation vector given by
(\ref{C5}) as described in Example 4 on tontine funds. 

Apart from the case of an active administrator, we also considered RS schemes
$\left(  \boldsymbol{\pi},\boldsymbol{P}\right)  $ with a passive
administrator. In this case, the investment vector is given by $\boldsymbol{\pi
}=\left(  \pi_{1},\pi_{2},\ldots,\pi_{n},0\right)  $, while the relative compensation
vector is given by $\boldsymbol{W}=\left(  W_{1},W_{2},\ldots,W_{n},0\right)
$, with
\begin{equation}
W_{i}=
\left(  \sum_{j=1}^{n}\pi_{j}\right)  \times P_{i}+\pi_{i}\times
P_{n+1}\qquad i=1,2,\ldots,n
\end{equation}
where $\boldsymbol{P}=\left(  P_{1},P_{2},\ldots,P_{n+1}\right)  $ is again
the relative compensation vector of the RS scheme under consideration. This setup is a
generalization of the setup in Denuit and Robert (2025), who consider RS
schemes with relative compensation vector given by (\ref{C5}) as described in Example 4 on tontine funds.

For both types of compensation-based RS schemes, we derived actuarial fairness
conditions, that is, conditions under which the initial investment of each
agent is equal to the expected compensation he will receive at time $1$.\ 

We considered two particular cases of tontine funds in some detail. First,
the two participant tontine fund was investigated. In this case, two participants
decide to set up a tontine fund and share the terminal fund value among the
survivors. We also considered the homogeneous tontine fund, where all
participants have i.i.d. survival indicator variables and purchase equal numbers of protection units. Finally, we also considered the relation between the homogeneous tontine fund approach with passive and active administrators and the centralized insurance approach. Addressing moral hazard in compensation-based risk-sharing schemes remains an open problem. Example 1 shortly discusses one possible approach, although a more thorough investigation of this issue is left for future research.

\paragraph{Acknowledgments.}
Jan Dhaene gracefully acknowledges funding from FWO and F.R.S.-FNRS under the Excellence of Science (EOS) programme, project ASTeRISK (40007517). Atibhav Chaudhry gracefully acknowledges support by an Australian Government Research Training Program (RTP) Scholarship, a Faculty of Business and Economics Doctoral Program Scholarship, a University of Melbourne Research Scholarship, and the 2021 Kilmany Scholarship. Ka Chun Cheung is supported by a grant from the Research Grants Council of the Hong Kong Special Administrative Region, China (Project No. 17303721). Austin Riis-Due gratefully acknowledges the support of The James C. Hickman Scholars Award. The authors also thank Michel Denuit, Runhuan Feng, Zinoviy Landsman, Daniël Linders, Moshe Milevsky and Bertrand Tavin for helpful comments on earlier versions of the paper.

\paragraph{Competing interests.}
The authors have no competing interests to declare.

\section{References}
\begin{itemize}
\renewcommand\labelitemi{\rule[0.5ex]{0.5em}{0.4pt}}
    \item Abdikerimova, S., \& Feng, R. (2022). Peer-to-peer multi-risk insurance and mutual aid. European Journal of Operational Research, 299 (2), 735–749.
    \item Bernard, C., Feliciangeli, M., \& Vanduffel, S. (2025). Can an actuarially unfair tontine be optimal? The Geneva Risk and Insurance Review, 50 (1), 39–71.
    \item Canabarro, E., Finkemeier, M., Anderson, R. R., \& Bendimerad, F. (2000). Analyzing insurance‐linked securities. The Journal of Risk Finance, 1(2), 49-75.
    \item Chen, P., Cheung, K. C., Dhaene, J., \& Yam, P. (2026). Optimal tontine fund design via rank-dependent expected utility. Working paper.
    \item Cheung, K. C., \& Lo, A. (2014). Characterizing mutual exclusivity as the strongest negative multivariate dependence structure. Insurance: Mathematics and Economics, 55, 180–190.
    \item Denuit, M., \& Dhaene, J. (2012). Convex order and comonotonic conditional mean risk sharing. Insurance: Mathematics and Economics, 51 (2), 265–270.
    \item Denuit, M., Dhaene, J., Ghossoub, M., \& Robert, C. Y. (2025). Comonotonicity and pareto optimality, with application to collaborative insurance. Insurance: Mathematics and Economics, 120, 1–16.
    \item Denuit, M., Dhaene, J., \& Robert, C. Y. (2022). Risk-sharing rules and their properties, with applications to peer-to-peer insurance. Journal of Risk and Insurance, 89 (3), 615–667.
    \item Denuit, M., \& Robert, C. Y. (2023). Endowment contingency funds for mutual aid and public financing. LIDAM Discussion Papers ISBA 2023/09, Université catholique de Louvain, Institute of Statistics, Biostatistics and Actuarial Sciences (ISBA).
    \item Denuit, M., \& Robert, C. Y. (2025). Equal compensations under actuarially fair contributions in endowment contingency funds. Risk Sciences, 1, 100005.
    \item Denuit, M., \& Robert, C. Y. (2026). Fully-funded risk-sharing schemes. LIDAM Discussion Papers ISBA 2026/14, Université catholique de Louvain, Institute of Statistics, Biostatistics and Actuarial Sciences (ISBA).
    \item Devroye, L. (1986). Non-uniform random variate generation. Springer-Verlag.
    \item Dhaene, J., \& Denuit, M. (1999). The safest dependence structure among risks. Insurance: Mathematics and Economics, 25 (1), 11–21.
    \item Dhaene, J., Kazzi, R., \& Valdez, E. A. (2026). Axiomatic characterizations of some simple risk-sharing rules. Risk Sciences, 2.
    \item Dhaene, J., \& Milevsky, M. A. (2024). Egalitarian pooling and sharing of longevity risk aka can an administrator help skin the tontine cat? Insurance: Mathematics and Economics, 119, 238–250.
    \item Dhaene, J., Robert, C. Y., Cheung, K. C., \& Denuit, M. (2025). An axiomatic characterization of the quantile risk-sharing rule. Scandinavian Actuarial Journal, 1–20.
    \item Jiao, Z., Kou, S., Liu, Y., \& Wang, R. (2022). An axiomatic theory for anonymized risk sharing. arXiv preprint arXiv:2208.07533.
    \item Lauzier, J.-G., Lin, L., \& Wang, R. (2024). Negatively dependent optimal risk sharing. arXiv preprint arXiv:2401.03328.
    \item Tavin, B. (2023). Reply to request from Dhaene and Milevsky [Private communication]
\end{itemize}

\end{document}